\documentclass[review, 3p, times]{elsarticle}
\usepackage{amsmath}
\usepackage{amsfonts}
\usepackage{amssymb}
\usepackage{amsthm}
\usepackage{bm}
\usepackage{indentfirst}
\usepackage{titlesec}
\usepackage{graphicx}
\DeclareGraphicsExtensions{.eps}
\usepackage{subfigure}
\usepackage{array}
\usepackage{fullpage}
\usepackage{color}
\usepackage{url}

\DeclareMathAlphabet{\mathsfsl}{OT1}{cmss}{m}{sl}

\newcommand{\PreserveBackslash}[1]{\let\temp=\\#1\let\\=\temp}
\newcolumntype{C}[1]{>{\PreserveBackslash\centering}p{#1}}
\newcolumntype{R}[1]{>{\PreserveBackslash\raggedleft}p{#1}}
\newcolumntype{L}[1]{>{\PreserveBackslash\raggedright}p{#1}}

\renewcommand{\theequation}{\thesection.\arabic{equation}}
\numberwithin{equation}{section}
\newtheorem{thm}{Theorem}[section]
\newtheorem{lem}[thm]{Lemma}

\theoremstyle{definition}

\newtheorem{rem}[thm]{Remark}

\begin{document}

\begin{frontmatter}

\title{RotEqNet: Rotation-Equivariant Network for Fluid Systems with Symmetric High-Order Tensors}

\author[purdue1]{Liyao Gao}
\author[jhu]{Yifan Du\fnref{fn1}}
\author[alec]{Hongshan Li\fnref{fn1}}
\author[purdue2]{Guang Lin\corref{gl}}

\fntext[fn1]{Equal contribution.}  
\address[purdue1]{Department of Mathematics, Purdue University, West Lafayette IN 47907, USA}
\address[jhu]{Department of Mechanical Engineering, John Hopkins University, Baltimore MD 21218, USA}
\address[alec]{Alectio, Santa Clara CA 95054, USA}
\address[purdue2]{Department of Mathematics and School of Mechanical Engineering, Purdue University, West Lafayette IN 47907, USA}
\cortext[gl]{Corresponding author.}

\ead{guanglin@purdue.edu}

\begin{abstract}
In the recent application of scientific modeling, machine learning models are largely applied to facilitate computational simulations of fluid systems. Rotation symmetry is a general property for most symmetric fluid systems. However, in general, current machine learning methods have no theoretical guarantee of Rotation symmetry. By observing an important property of contraction and rotation operation on high order symmetric tensors, we prove that the rotation operation is preserved via tensor contraction. Based on this theoretical justification, in this paper, we introduce Rotation-Equivariant Network (RotEqNet) to guarantee the property of rotation-equivariance for high order tensors in fluid systems. 
We implement RotEqNet and evaluate our claims with four case studies on various fluid systems. The property of error reduction and rotation-equivariance is verified in these case studies. Results are showing the high superiority of RotEqNet compared to traditional machine learning methods.  
\end{abstract}


\end{frontmatter}

\section{Introduction}
With recent developments in data science and computational tools, machine learning algorithms have been increasingly applied in different engineering and science areas to model physical phenomena. The data from physical experiments and numerical simulations are a source of knowledge about the physical world, on which data-driven methods could be performed to extract new physical laws \citep{huang2020data,raissi2017physics,carleo2019machine,kutz2017deep,wang2017physics,li2019accelerating}. 
For example, in turbulence RANS modeling in fluid mechanics, traditional modeling methods have failed in many flow scenarios. A unified RANS model that can successfully describe complex flows, including boundary layer, a strong rotation, separation still does not exist according to the author's knowledge \citep{durbin2018some,ling2015evaluation}. On the other hand, advanced measurement and direct numerical simulations provide plenty of data that could be utilized to establish and validate new models. With the above argument, data-driven methods are particularly suitable for turbulence modeling and some other areas in physics and engineering. There have been many attempts to discover new turbulence models using machine learning methods. Milano and Koumoutsakos \citep{milano2002neural} reconstruct near-wall flow applying neural networks and compared their results with linear methods (POD). Zhang and Duraisamy \citep{zhang2015machine} used Gaussian process regression combined with an artificial neural network to predict turbulent channel flow and bypass transition. Beck, Flad, and Munz \citep{beck2019deep} applied residual neural network for Large Eddy Simulation. Chen et. al. proposed an ODE network to generally learn differential equations \citep{chen2018neural}.

The physical laws often appear in the form of tensorial equalities which inherently obey certain types of symmetry. For example, the constitution laws in fluid and solid mechanics should obey translation and rotation invariance \citep{mase2009continuum}. The turbulence RANS model is local tensorial equality between mean velocity gradient and Reynolds stress. The turbulence RANS models should also be rotation invariant \citep{pope2001turbulent,pope1975more}. However, machine learning methods for RANS modeling do not automatically guarantee rotation invariance, if we use Cartesian components of tensors as input and output of training data. This problem has been addressed by \citep{ling2016machine,wang2017physics}. In \citep{ling2016machine,ling2016reynolds}, Reynolds stress is expressed as a general expansion of nonlinear integrity basis multiplied by scalar functions of invariants of strain rate and rotation rate tensors. Machine learning is performed to find these scalar functions of tensor invariants of strain rate and rotation rate tensors. Mathematically this expansion comes from an application of the Caylay-Hamilton theory. The special case used in \citep{ling2016machine,ling2016reynolds} is derived by S.B.Pope in \citep{pope1975more}. Although such construction is general and possible for higher-order tensors and tensor tuples containing multiple tensors, the number of this basis and the derivation complexity will grow exponentially and become prohibitive for real applications \citep{johnson2016handbook,smith1965isotropic}. 

Why would this problem of rotation-equivariance be hard to solve? At first glance, if a system has the property of rotation-equivariance, one has more information for this system. Therefore, this added property of rotation-equivariance would lower the performance of a learner. More specifically, adding this new rule of rotation symmetry in a system will require the machine learning algorithm to extract more rules from existing data \citep{mohri2018foundations}. 
In this case, the property of rotation-equivariance could be considered as a continuous group action. 
There is limited research in the field of deep learning that considers the preservation of symmetries under continuous group actions for physical systems.
To address our second point, continuous information is hard to be absorbed. If we consider a machine learning algorithm as an information compression process from input to output \citep{saxe2019information}, a continuous transformation as rotation will be difficult for learning algorithms to absorb.

Given the universal approximation theorem by \citep{hornik1989multilayer}, it would seem that the application of neural networks, especially deep neural networks could solve any problem. As formulated by \citep{muller1999application,weatheritt2017comparative,qin2019data,han2019solving}, advanced machine learning methods, especially deep neural networks \citep{lecun2015deep}, seem to provide a new opportunity for physical equations approximation. However, in this case of rotation symmetry, if we use a multiple layer perceptron $M$ to learn the relation $f$, then most likely $M$ does not preserve rotation-equivariance. Generally, the neural network function classes do not satisfy rotation equivariance.

There have been previous works considering group-equivariance with convolutional neural networks in image recognition. A general method has been proposed using group convolution \citep{cohen2016group,esteves2020theoretical,esteves2019equivariant}. Based on the idea of using convolution, several methods composed a steerable filter for rotation-equivariance in convolutional neural networks \citep{weiler2018learning,cheng2018rotdcf,finzi2020generalizing,gao2019rotation}. However, these works cannot be applied in physical systems as well. One of the most important reasons is that the rotation operation on the image is different from rotation operation on physical systems. Consider a rotation operation on a specific image. We are thinking of a transformation from polar coordinates centering at a certain point \citep{foley1996computer}. This kind of transformation is different from rotation operation on tensors. Additionally, these methods have a strong restriction that this model must be built on convolutional neural networks. Yet, considering physical systems, convolutional neural networks might not be the best choice since they are designed for image processing.

The problem of rotation-equivariance is also quite impossible to be simply solved by data augmentation and preprocessing. Mentioned by previous works \citep{ling2016machine}, a typical solution is to apply the technique of data augmentation. However, the method of data augmentation fails to have a theoretical guarantee of obtaining the property of rotation-equivariance with finite sample set.
Data augmentation method has a theoretical foundation that at infinite sample limit it will asymptotically reach rotation equivariance.
However, such a dataset is not only difficult to obtain but also requires much higher computation power while training the model. In the case of using naive preprocessing methods, the problem is that there are limited theoretical tools to deal with high-order tensors, and only limited methods to use for low order tensors. It is hard to apply specific techniques, such as diagonalization, in the case of high-order tensors. Since naive data preprocessing methods are impossible to apply, a more complex method with a theoretical guarantee should be proposed in order to solve this problem.

In this paper, we establish Rotation-Equivariant Network (RotEqNet), a new data-driven framework, which guarantees rotation-equivariance at a theoretical level. 
Different from previous methods, we first find a method to preserve rotation operation via tensor contraction. In our proposed position standardization algorithm, it could properly link a high-order tensor to a low order tensor with the same rotation operation. 
By applying mathematical tools for low order matrices (diagonalization and QR factorization), a desired standard position could be derived by the rotation matrix from the previous step. Standard position algorithm is proven to be rotation-invariant in Theorem \ref{thm:1}, \emph{i.e.} two tensors differ by a rotation would have the same standard position. Therefore, the learning rules based on standard position are forming a quotient space of the original rules in random rotated plural position \citep{weiler2018learning,zhou2019continuity}. In this way, RotEqNet lowers the training difficulty of a randomly positioned dataset. Further, RotEqNet is also proven to be rotation-equivariant, as we have shown in Theorem \ref{thm:2}. These advantages of RotEqNet would result in an observable error reduction compared to previously introduced data-driven methods. 
We applied RotEqNet into four different case studies ranging from second-order, third-order, and fourth-order. These case studies are designed based on Newtonian fluids, Large-eddy simulations, and Electrostriction. Improved performances could be observed for using RotEqNet. The error is reduced for 99.6\%, 15.62\%, and 54.63\% for second, third, forth-order case studies, respectively. 
Our contribution in this paper is three-fold: 
\begin{enumerate}
  \item We showed an important property of contraction operation on tensors. Contraction operation will preserve rotation operation on tensor with arbitrary order. This is stated in Lemma \ref{lemma2.3}. 
  \item We propose a properly designed RotEqNet with a position standardization algorithm to guarantee the property of rotation-equivariant. We proved the property of rotation-invariant of position standardization algorithm in Theorem \ref{thm:1} and the property of rotation-equivariant of RotEqNet rigorously in Theorem \ref{thm:2}.
  \item We implement our proposed algorithm and the architecture of RotEqNet. We further conduct case studies to show its credibility in design and superiority compared to baseline methods. 
\end{enumerate}

To provide a general architecture of our paper, in Section \ref{sec:preliminaries} we introduce basic definitions of rotation for arbitrary order tensor (tuples) and related concepts. In Section \ref{pre:rotEq} we formulate rotation invariance (equivariance) on supervised learning methods. The RotEqNet and main algorithm is presented in Section \ref{sec:roteqnet}, and numerical results are shown in Section \ref{sec:expr}. 

\section{Preliminaries and Problem Description}
\label{sec:preliminaries}
\subsection{Tensor and its operations}
\label{sec:tensorDef}
In this section, we first introduce an abstract way of defining tensor. One reason for us to introduce the more abstract way to think about tensors is that it provides a convenient formalism for the operations we will do on the tonsorial data discussed in the previous section. The operations are 

\begin{enumerate}
\item Linear transformation
\item Contraction
\end{enumerate}

The formalism helps us to prove that these two operations commute which lays 
theoretical ground for the computation of a representative of rotationally-relatated
tensors. We will call this representative \emph{standard position}

\subsubsection{Abstract definition of tensors}
Following \cite{curtis2012abstract}, fix a vector space $V$ of dimension $n$ over $\mathbb{R}$. A \emph{tensor product}
$V\otimes V$ is a vector space with the property that $\mathbb{R}$-bilinear maps
$V \times V \rightarrow \mathbb{R}$ are in natural one-to-one correspondence
with $\mathbb{R}$-linear maps $V \otimes V \rightarrow \mathbb{R}$. 

The tensor product $V\otimes V$ can be constructed as the 
quotient vector space $V\times V / C$, where $C$ is generated by 
vectors of the following types
\begin{equation}
    (i)\;(x+y, z) - (x, z) - (y, z) \\
    (ii)\;(x, y+z) - (x, y) - (x, z) \\
    (iii)\;(ax, y) - a(x, y) \\
    (iv)\;(x, ay) - a(x, y) \\
\end{equation}
where $x$ and $y$ are vectors in $V$ and $a$ is a scalar in $\mathbb{R}$. This means
any element in $C$ can be written as a linear combination of vectors of the 
above form. $C$ is not necessarily a vector space of finite dimension. 
But the quotient space $V\otimes V$ is. Let $g: V \times V \rightarrow V\otimes V$
be the natural projection map, then we use $x\otimes y$ to denote the image of $(x, y)$
under $g$. 

Let $\langle e_1,\cdots, e_n\rangle$ be a basis of $V$, then $e_i\otimes e_j$ for 
$i=1,...,n$ and $j=1,...,n$ form a basis of $V\otimes V$. 
This means any vector $p \in V\otimes V$ can be written as 
\begin{equation}
    \sum_{i,j}a_{ij}e_i\otimes e_j
\end{equation}
for some $a_{ij} \in \mathbb{R}$. 

Here are some relations of tensors which come directly as a consequence
of the relations generating $C$:
\begin{equation}
    a(e_i\otimes e_j) = a e_i \otimes e_j = e_i \otimes a e_j
\end{equation}

\begin{equation}
    (a_i e_i + a_j e_j)\otimes (a_k e_k) = a_ia_k(e_i\otimes e_j) + 
    a_ja_k (e_j\otimes e_k)
\end{equation}

The representation of a tensor in $V\otimes V$ is similar to the representation
of a linear map $V \rightarrow V$, i.e. a matrix. In fact, there is a natural way
to think of a tensor as a linear map: 

For each element $e_i\otimes e_j$ in the basis of $V\otimes V$, we can think of it as 
a linear map $V \rightarrow V$ by defining $e_i\otimes e_j (v) = e_i<e_j,v>$, where 
$<,>$ is the natural inner product on $V$. 
Extend the definition linearly to every element in $V\otimes V$, we obtain a way to 
identify $V\otimes V$ as the space of linear map $V \rightarrow V$. In fact, the 
tensor $\sum_{i,j}a_{i, j} e_i\otimes e_j$ corresponds to the linear map represented
by the matrix $[a_{ij}]$. 

We have defined the tensor product $V\otimes V$ over $V$. The definition/construction of order $k$ tensor $\overbrace{V\otimes\cdots\otimes V}^{k}$ follows the same course. We will denote order $k$
tensor by $\otimes^k V$.

The basis of $\otimes^k V$ is given by $e_{i1}\otimes\cdots\otimes e_{ik}$, where 
$i=1,...,n$ and $j=1,...,k$. With respect to this basis, any order $k$ tensor can
be written as $\sum_{i1,\cdots,ik}a_{i1,...,ik}e_{i1}\otimes\cdots\otimes e_{ik}$. 
Analogous to the order 2 case, we can think of an order $k$ tensor as a $k$-dimensional
matrix, the typical way a tensor in physical experiments are represented.

We will use $T^k$ to denote a tensor of order $k$, i.e. a vector in $\otimes^kV$.
$k$ is called the rank of the tensor.

\subsubsection{Rotation on tensors: a linear transformation}
A linear transformation on higher-order tensor is a generalization of a linear transformation on the first-order tensor, i.e. a vector.

Let $g: V \rightarrow V$ be a linear transformation. Use the basis 
$\langle e_1, \cdots, e_n\rangle$ of $V$, we can represent this expression with the equation
\begin{equation}
    g(e_i) = \sum_{j=1}^n a_{ij}e_j
\end{equation}
Let $M(g)$ denote the matrix representation of $g$ with respect to
the basis $\langle e_1, \cdots, e_n\rangle$. Then
\begin{equation}
    M(g) = [a_{ij}]^t
\end{equation}
i.e. the transpose of the matrix $[a_{ij}]$

The map $g$ naturally induces a map $\otimes^k g$ on $\otimes^k V$. 
On the basis element $e_{i1}\otimes\cdots\otimes e_{ik}$, the action of 
$\otimes^k g$ is defined as 
\begin{equation}
    e_{i1}\otimes\cdots\otimes e_{ik} \mapsto 
    g(e_{i1}) \otimes\cdots\otimes g(e_{ik})
\end{equation}
For any tensor $T \in \otimes^kV$, we will use $g(T)$ to denote 
the extension of $g$ on $\otimes^k V$

There is a convenient way to represent a linear transformation
of 2-tensor as matrix multiplication. 

For a 2-tensor $T = \sum_{i, j}b_{ij} e_i\otimes e_j$, use
$M(T)$ be the matrix whose $(i, j)$ term is $b_{ij}$. 
\begin{lem}
Rotation operation by matrix $R$ on second-order tensor (matrix) is a change of basis operation. 
\begin{equation}
    M(R(T)) = M(R)\times  M(T) \times M(R)^t,
\end{equation}
\label{lemma2.1}
where $\times$ here means the usual matrix multiplication. 
\end{lem}

\begin{rem}
Rotation operation by matrix $R$ on first-order tensor (vectors) $T$ could be viewed as
\begin{equation}
    M(R(T)) = M(R)\times  M(T).
\end{equation}
\label{lemma2.2}
\end{rem}

The proof here of Lemma \ref{lemma2.1} and Remark \ref{lemma2.2} are left in \ref{appendix:tensor}. 

Lemma \ref{lemma2.1} and remark \ref{lemma2.2} will be used in the proof of Theorem \ref{thm:1}. As we have shown in this subsection, one could use a matrix form of rotation operation with certain rules of matrix multiplication to perform a rotation on the tensor. In the following proofs of this paper, we applied this idea to perform rotation operation on tensors via matrix multiplication. 
\subsubsection{Contraction on tensors: reduction of order}
Let $\langle, \rangle$ be the standard inner product on $V$. Using this inner product, we can define the contraction of a tensor. It "merges" vectors on the specified axes using the inner product and reduces the rank of the tensor by 2. Formally, let $C(a, b)$ denote the 
contraction along axis-$a$ and axis-$b$. Here, the axis means the ordinal of $V$ in $\otimes^kV$. For example, axis-$1$ refers to the first copy of $V$ in 
$\otimes^kV$. 

On the element $\otimes_{j=1}^k v_{ij}$, $C(a, b)$
acts on it by 
pairing $v_{ia}$ and $v_{ib}$ via the inner product $\langle, \rangle$, i.e.
\begin{equation}
    C(a, b)(v_{i1}\otimes\cdots\otimes v_{in}) = \langle v_{ia}, v_{ib}\rangle
    v_{i1}\otimes\cdots \check{v_{ia}}\cdots\check{v_{ib}}\cdots\otimes v_{in}
\end{equation}
where $\check{v}$ means $v$ is not present. 

We can then define $C(a, b)$ on $\otimes^kV$ by extending linearly. 
When $k=2$, contraction is nothing other than taking the trace of the 
corresponding matrix.

\begin{lem}
Let $R: V \rightarrow V$ be a rotation. Let $T \in \otimes^k V$, then
\begin{equation}
    C(a,b)(R(T)) = R(C(a,b)(T))
\end{equation}
\label{lemma2.3}
\end{lem}
Lemma \ref{lemma2.3} shows an interesting connection between rotation operation and contraction. To understand this lemma, it represents that the contraction of a tensor is compatible with a linear transformation if this linear transformation is a rotation. This is an important lemma which is the foundation of the entire analysis in this paper. We would further utilize this lemma for extracting its rotation operation from higher (arbitrary) orders. We show the proof in \ref{appendix:tensor}. 

\label{sec:prob}
\subsection{Supervised learning setup}
In our problem, given data set $\mathcal{D}=\{X_i;y_i\}_{i=1, ..., N}$. The data set contains $N$ input-output pairs $(X_i;y_i)$. The input here is a tensor tuple: 
\begin{equation}
    X_i=[X_1, X_2, ..., X_{N_x}]
\end{equation}
$N_x$ is the length of $X_i$. Normally, we only have one output.

Generally speaking, following the definition of \cite{bishop2006pattern,tao2005supervised}, parametric supervised learning can be viewed as a type of a model composed from two parts. The first part is a 
predictor. Given parameter $\theta$, we have:
\begin{equation}
    \hat{y}=\mathcal{M}^{\theta}(X_i)
\end{equation}

, where $\hat{y}$ is the prediction output of learning model $M$, $\theta$ is the parameter of $M$. As stated, it predicts value based on input $X_i$. 

The second part is an optimizer, which updates the parameter $\theta$ based on a loss function. For a regression model, a typical loss function would be defined as:
\begin{equation}
    L(M, \theta)=\frac{1}{N} \sum_{i=1}^{N} \|y_i-M^{\theta}(X_i)\|^2,
\end{equation}
where $\|\cdot\|$ represents 2-norm. 

We usually hopes to minimize this loss function. It is formulated by: 
\begin{equation}
    \hat{\theta} = \arg\min_{\theta} L(M, \theta)
\end{equation}
where $\mathcal{M}$ is a learning model and $\mathcal{M}^{\hat{\theta}}$ is the optimal solution. Specifically, in this work, we applied Neural Networks \citep{specht1991general} and Random Forests \citep{liaw2002classification} in the case studies.

\subsection{Obtaining rotation-equivariance properties in systems using supervised learning}
\label{pre:rotEq}

Group equivariance is an important property for most physical systems. Typical examples of group equivariance could be rotation group equivariance, scaling group equivariance, and translation group equivariance. Mathematically, group equivariance is a property of a mapping $f:X\rightarrow Y$ to commute from $X$ to $Y$ under rotation group actions. Specifically, let $R\in SO(n)$ be a rotation action. $f:X\rightarrow Y$ is rotation-equivariant if
\begin{equation}
    f(R(x))=R(f(x)), \;\;\;\forall R \in SO(n),\: x \in X.
\end{equation}
As a special case of rotation-equivariant, a function $f:X\rightarrow Y$ is rotation-invariant if:
\begin{equation}
    f(R(x))=f(x), \;\;\;\forall R \in SO(n),\: x \in X.
\end{equation}

Since supervised learning models could be considered as functions, name a machine learning model as $M^{\theta}$. For a rotation operation $R$, we hope to obtain the property that: 
\begin{equation}
    M^{\theta}(R (x))=R(M^{\theta}(x)), \;\;\;\forall R \in SO(n),\: x \in X
    \label{def:rotEq}
\end{equation}
For analysis below in Sec. \ref{sec:analysis}, we prove the rotation-equivariance property following the definition stated here in Equ. \ref{def:rotEq}. In other words, if a system would satisfy the property in Equ. \ref{def:rotEq}, then this system is rotation-equivariant.

\subsection{Modeling symmetric fluid systems via supervised learning}

The machine learning approach to the fluid dynamics modeling involves training a supervised learning model $\mathcal{M}$ using
$X_i$ as features and $Y$ as label. 

In our case, the underlying space $S$ of the fluid dynamic system is complete with respect to rotation. This means for all rotation 
$R: \mathbb{R}^n \rightarrow \mathbb{R}^n $, $R(p) \in F$ for all $p \in S$. 
The objects we want to model via machine learning are rotation-equivariant tensorial
fields on $S$. 

Let $X$ be a tensorial field on $S$, for any point $p \in S$, we use $X(p)$ to denote
the tensor at $p$ (for example, pressure at a particular point in a fluid dynamics system). $X$ is said to be \emph{rotation-equivariant} if for all point $p \in S$ and
all rotation $R$
\[
    X(R(p)) = R(X(p))
\]

Suppose one has tensorial fields $X_1,\cdots, X_n, Y$ on $S$ such that
$X_i$ and $Y$ are related by some unknown physical law $f$ such that
\[
    f(X_1, \cdots, X_n) = Y
\]

Supervised machine learning methods can be used here to learn a function $\mathcal{M}$
that approximates $f$ such that $\mathcal{M}$ generalizes well on new data.

Suppose those tensorial fields are rotation-equivariant, then naturally the model
$f$ as well its proxy $\mathcal{M}$

\section{Rotation Equivariant Network}
\label{sec:roteqnet}
In this section, we would like to propose Rotation Equivariant Network (RotEqNet) to solve rotation problems for high order tensors in fluid systems. RotEqNet is based on the position standardization algorithm, as we would further discuss in Section \ref{RotInvExtrAlgo}. We first provide a general description of the whole architecture in \ref{motivation}. 

\subsection{Architecture} \label{motivation}

As shown in Figure\ref{fig:RotEqNet}, RotEqNet generally goes through three important steps: position standardization, prediction of kernel predictor, and position resetting. To be specific, the position standardization is an algorithm to transfer incoming tensor to its standard position. In Figure\ref{fig:RotEqNet}, the 'even order standardization' and 'odd order standardization' sections denote this algorithm in position standardization. Then, $X_s$ is considered as a standard position of input tensor $X$, and $R$ is an extracted rotation operation to transfer between standard position and original position. The output of kernel predictor is only dealing with standard positions. This will result the output $y_s$ in its standard position as well. Finally, apply $R^{-1}$ to output $y_s$ will be our final prediction. A general mathematical description of this process could be described as: 

\begin{figure}
    \centering
    \includegraphics[scale=0.5]{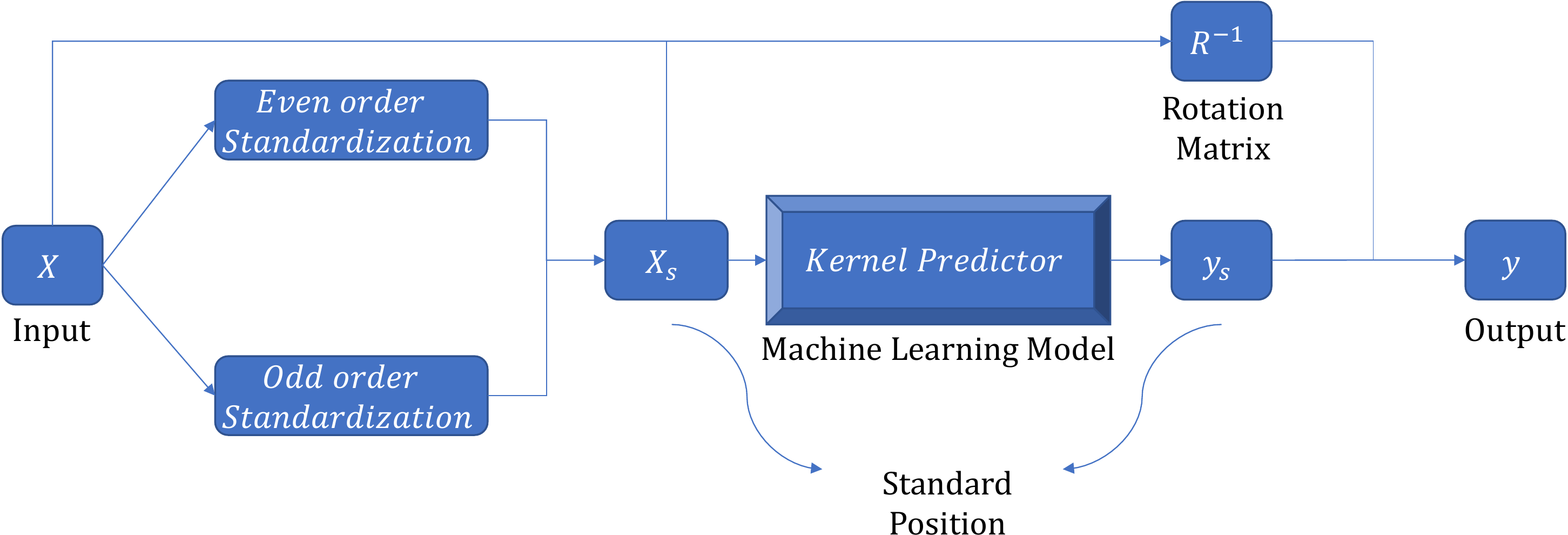}
    \caption{Rotation-equivariant Network Architecture}
    \label{fig:RotEqNet}
\end{figure}

\begin{equation}
    y=R(M^{\theta}(R^{-1}(T)))
\end{equation}

How would this process help to solve rotation problems for high order tensors? \label{reason}

An important reason is related to a reduced function space for learning. When a learning model is only training with the standard position, it would no longer still have to deal with the entire group action causing a group-equivariant, but only need to focus on the pattern by the related physical equation. 
Name the rotation group as G, and consider a full function space $\mathbf{C}(X,Y)$. As mentioned in \citep{weiler2018learning}, instead of performing regression on $\mathbf{C}(X,Y)$, RotEqNet is essentially exploring a much smaller space $\mathbf{C}(X/G,Y/G)$. The reduction of input-output dimensionality makes the training easier. With the same number of samples, the pattern for learning requires a far smaller space to explore. 
The second reason is RotEqNet could provide a theoretical guarantee of the property of rotation-equivariant. Utilizing rotation symmetry as a strong prior for most physical systems, RotEqNet have a better generalized result learning from limited amount of data. 

The following subsections will introduce position standardization algorithm in a complete manner with proof on its property of rotation-invariant in Theorem \ref{thm:1}. Then, we will demonstrate the proof of showing RotEqNet is rotation-equivariant in Theorem \ref{thm:2}.

\subsection{Position standardization algorithm for High Order Tensors}

\label{RotInvExtrAlgo}
Let $\mathcal{D}$ denote our data set. 
The first stage of RotEqNet is to find a good representative of all tensors that are related to each other by rotation. We will call this representative the sample in "standard position," and we will denote it by $(X_s, Y_s)$. We will use $S$ to denote the position standardization algorithm and $S(X, Y) = (X_s, Y_s)$ to mean reducing $(X, Y)$ to its standard position.

$S$ has the following property that $\forall (X, Y) \in \mathcal{D}$ and all rotation operation $R\in SO(n)$,
\begin{equation}
    S(R(X), R(Y)) = (X_s, Y_s).
\label{def:standardpos}
\end{equation}
This means, 
$S$ produces exactly the same output no matter how $(X, Y)$ is 
rotated, \emph{i.e.} it is rotation-invariant.

Intuitively, for a tensor $T$, we are selecting a representative on the orbit $O(T)$, (where $O(T)=\{R\cdot T|R\in SO(n)\}$), as the rotation invariant of a $T$ \citep{pinter2010book}. In our algorithm, we initially perform a tensor contraction to higher-order tensors, reducing the dimension to obtain a lower order tensor. Then using diagonalization for even cases and QR factorization for odd cases, the algorithm could obtain a rotation operation acting on $T$. Finally, it could get a tensor in standard position by rotating $T$ the original tensor with the inverse of the obtained rotation matrix.

This operation is compatible with the theoretical result shown in Lemma \ref{lemma2.3}.

\subsubsection{Tensor of even order}
\begin{figure}
    \centering
    \includegraphics[scale=0.6]{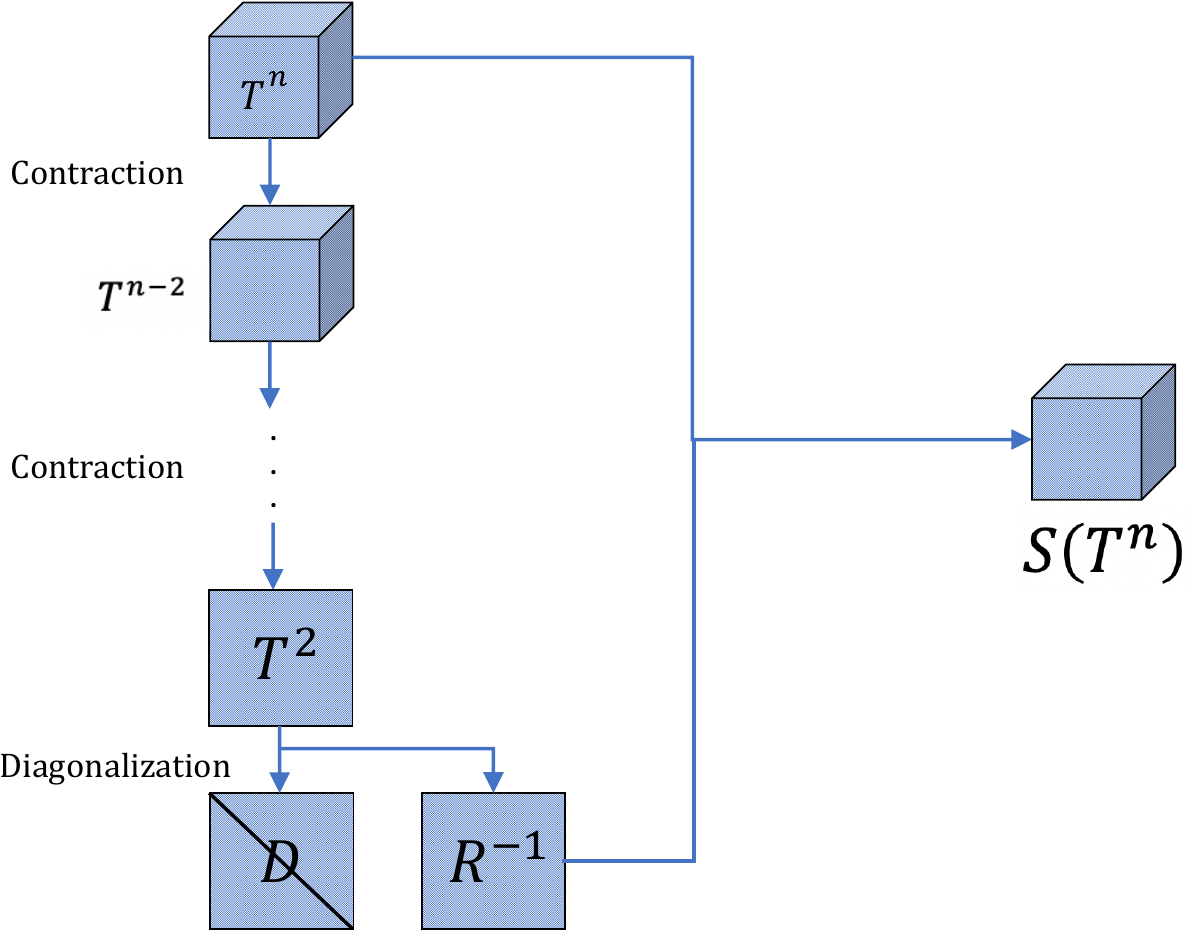}
    \centering
    \caption{Rotation-invariant extraction for even order.}
    \label{fig:evenorder}
\end{figure}
Given a symmetric tensor of even order $T^n \in \otimes^n V$ ($n$ is even). Let $\mathcal{C}$ denote a sequence of contraction
along the first two axes until we reach a second-order tensor. Applying $\mathcal{C}$ to $T^n$ we get:
\begin{equation}
T^2 = \mathcal{C}(T^n)
\end{equation}

\begin{align}
T^n \underrightarrow{\;\;\mathcal{C}\;\;} T^{n-2}\underrightarrow{\;\;\mathcal{C}\;\;} (...) \underrightarrow{\;\;\mathcal{C}\;\;}T^2
\end{align}

Then we find the orthonormal eigen-vectors of $T^2$ and use them to form the 
orthonormal matrix $R$ that diagonalize $T^2$
\begin{equation}
    T^2 = R^{-1} \times D \times R
\end{equation}
Since $R$ is an orthonormal matrix, we have
\begin{equation}
    R^{-1} = R^T
\end{equation}

We will call $D$ the standard position of $T^2$.
We write $R(T^2) = D$ to shorten the notation

Since contraction and rotation are compatible by Theorem \ref{lemma2.3}. We can apply $R$ to $T^n$ before we apply contraction, and we will have
\begin{equation}
\mathcal{C}(R(T^n)) = D
\end{equation}

For the even tensor $T^n$, we define 
\begin{equation}
    S(T^n) = R^{-1}(T^n)
\end{equation}

\subsubsection{Tensor of odd order}

Why would it be different for even order and odd order? Since odd order cannot directly reduce its dimension to 2 by contraction. Due to the fact that each contraction will reduce the dimensionality by 2, the reduced dimension will also be an odd number, which cannot be 2. Involving in this problem, this would further be impossible to extract the rotation matrix, which is impossible to rotate the tensor into a standard position. The following described is the method that we use to solve the problem.

Given a symmetric tensor of an odd order tensor $T^n \in \otimes^n V$ ($n$ is odd). Let $\mathcal{C}$ denote a sequence of contraction
along the first two axes until we reach a third-order tensor. Applying $\mathcal{C}$ to $T^n$ we get:
\begin{equation}
T^3= \mathcal{C}(T^n)
\end{equation}
\begin{figure}
    \centering
    \includegraphics[scale=0.65]{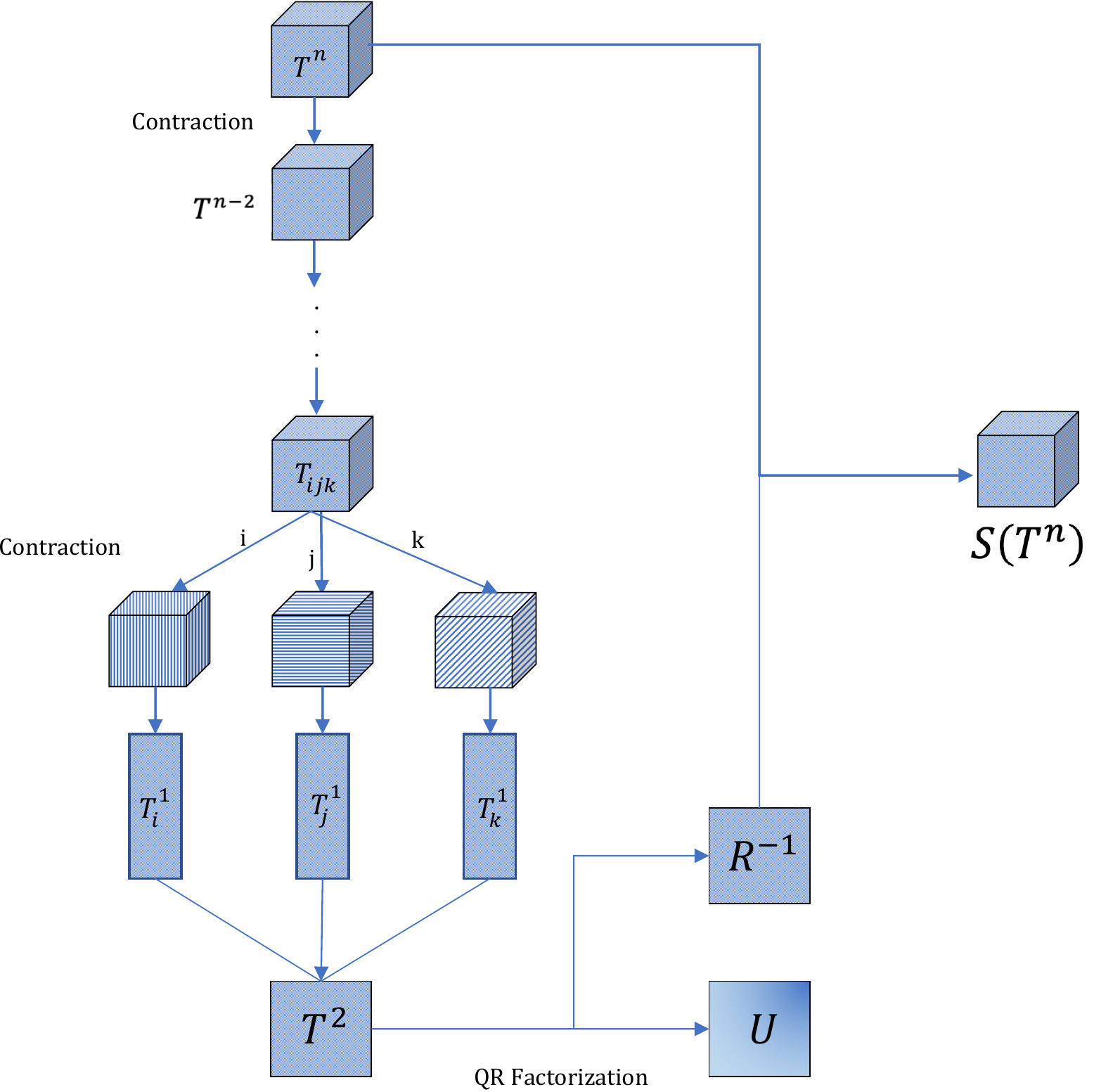}
    \caption{Rotation-invariant extraction for odd order.}
    \label{fig:my_label}
\end{figure}

After we obtain $T^3$, we could obtain three different order one tensors by contracting it. Name the contracted results, which are first-order tensors i.e. vectors, $V_1, V_2, V_3$. We could get an order two tensor by concatenating them.

\begin{equation}
T^2=(V_1, V_2, V_3)
\end{equation}

Then, we could perform the a similar process as before. We perform QR factorization to obtain rotation matrix $R$. 
\begin{equation}
T^2=R\times U^2, 
\label{QR factorization}
\end{equation}
For odd tensor, we define: 
\begin{equation}
    S(T^n) = R^{-1}(T^n)
\end{equation}

The pseudocode of our proposed algorithm is documented in Algorithm 1. We evaluate our method of position standardization algorithm in Section \ref{sec:expr}. 

\begin{figure}
    \centering
    \includegraphics[scale=0.95]{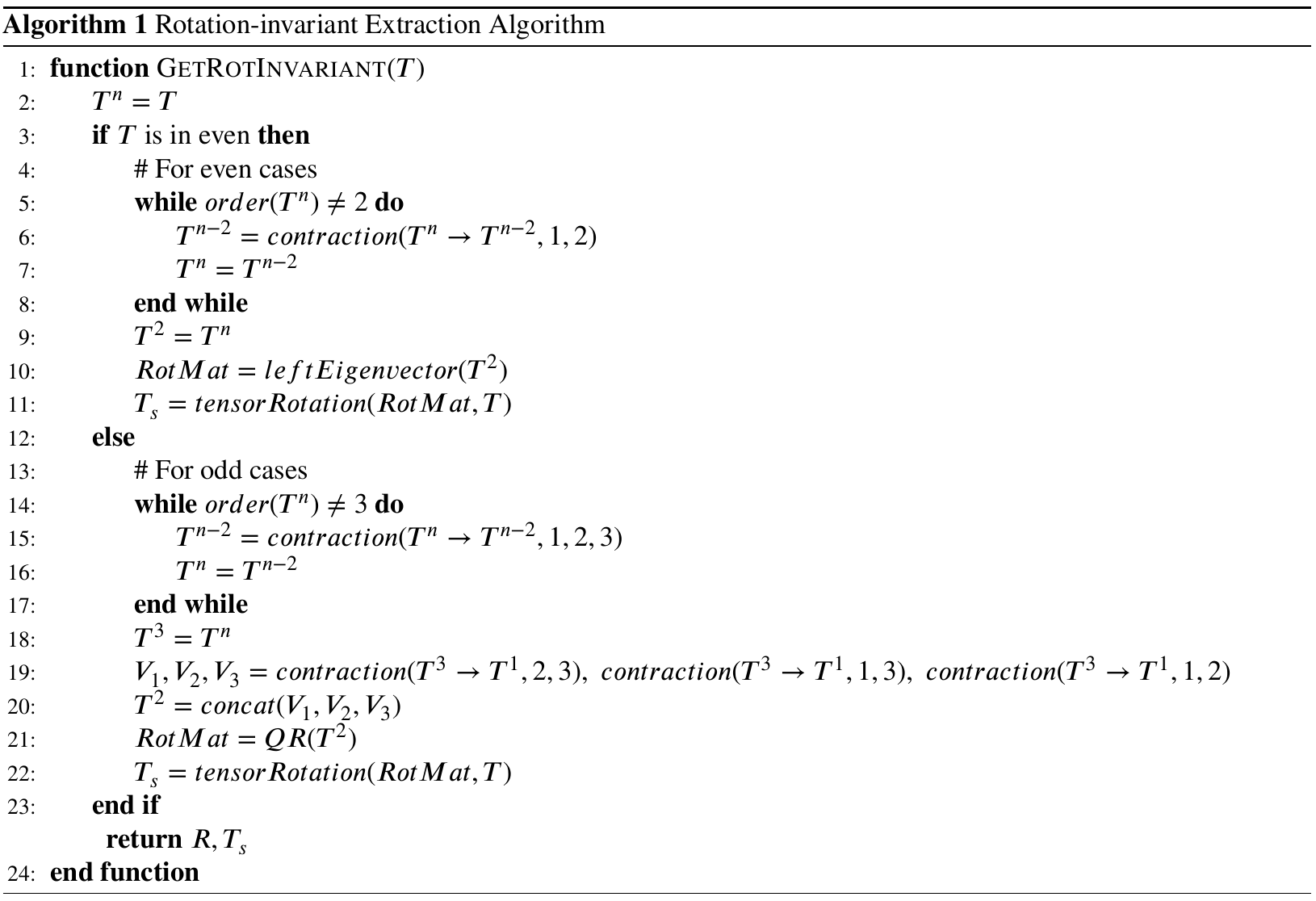}
    \label{rotexAlgo}
\end{figure}

\subsection{Theoretical Analysis of Rotation-equivariant property}

\label{sec:analysis}
In our analysis, we aim to show the rotation-equivariant property of RotEqNet. As an important first step, we need to analyze the property of rotation invariant (standard position) derived by the position standardization algoriTheorem We hope to show Equ. \ref{def:standardpos} is true. Once Equ. \ref{def:standardpos} is proved true, RotEqNet will automatically be rotation-equivariant based on its architecture. 

To sketch an outline about theorems below, Lemma \ref{lemma2.3} would serve as an important fact for preserving rotation information after contraction. Our algorithm has been analyzed by Theorem \ref{thm:1}. 

\subsubsection{Main theorems and proofs}

\begin{thm}
\label{thm:1}
$S$ is rotation invariant, \emph{i.e.} for all rotation $R \in SO(n)$ and 
symmetric tensor $T \in \otimes^k V$
\begin{equation}
    S(R(T)) = S(T)
\end{equation}
\end{thm}
We provide a proof in Appendix \ref{proof:thm2}. 
We call $S(T)$ the standard position of $S(T)$

Using Theorem \ref{thm:1} we can automatically have the result on tensor with arbitrary order that the standard position, derived by position standardization algorithm, is rotation invariant. 

\begin{thm}
RotEqNet, $M_R$, is rotation-equivariant, \emph{i.e.} for all rotation $R \in SO(n)$ and 
tensor $T \in \otimes^k V$
\begin{equation}
    M_R(R(T)) = R(M_R(T))
\end{equation}
\label{thm:2}
\end{thm}
\begin{proof}
Name RotEqNet as $M_{R}$, kernel classifier as $M_k$. Consider a input pair $(X, y) \in \mathcal{D}$. Suppose the result of standardize position algorithm $S$ would have $S(X)=P_1(X)$, where $P_1$ denote a rotation operation. 

First, consider $M_R(X)$ the process of RotEqNet is defined as:
\begin{equation}
    M_R(X)=P_1^{-1}( M_k(S(X)))=P_1^{-1}(M_k(P_1(X)))
\end{equation}
Consider another rotation operation $P_2$ in the matrix form acting on $X$, using Theorem. \ref{thm:1} we know that:
\begin{equation}
    S(P_2(X))=S(X)=P_1(X)=(P_1 I) (X)=(P_1 \times P_2^{-1})(P_2(X)),
    \label{thm4:1}
\end{equation}
where $I$ is an identity matrix. 

Then, consider $M_R(P_2(X))$ the process of RotEqNet is defined as:
\begin{equation}
    M_R(P_2(X))=(P_1 P_2^{-1})^{-1}(M_k(S(P_2(X))))=(P_2P_1^{-1}) (M_k(S(P_2 (X))))
\end{equation}
We know that $S(P_2(X))=S(X)$ from Equ. \ref{thm4:1}. Therefore, we have $M_R(S(X))=M_R(S(P_2(X)))$. Substitute $M_R(X)$ back into previous equation, 
\begin{equation}
    M_R(P_2(X))=P_2P_1^{-1}(M_k(S(X)))=P_2(M_R(X))
\end{equation}
To simplify, for rotation operation $P_2$ on input $X$, we have
\begin{equation}
    M_R(P_2(X))=P_2(M_R(X))
\end{equation}
This is showing that $M_R$ is rotation-equivariant by definition. Therefore, RotEqNet is rotation-equivariant. 
\end{proof}
it has shown that \textit{Algorithm 1} is able to preserve rotation information for low dimension, and further extract using diagonalization for matrices. This part is a theoretical guarantee of our position standardization algoriTheorem

\section{Case studies}
\label{sec:expr}

In this section, a series of cases are provided to show the performance of RotEqNet. In the following subsections, cases are included from second-order, third-order, and fourth-order. We specifically investigate second-order cases with detailed studies on linear, and nonlinear test equations since, in current applications, second-order physical systems are widely used. Generally, we reported two properties of RotEqNet in every case study. 
The first property is a loss reduction property. We apply RotEqNet in each test physical equation and compared it to the baseline models (Neural Networks and Random Forests). Another one is the rotation-invariant property. We examine this property by letting RotEqNet and baseline models to perform prediction on rotations of randomly selected data. 
We report detailed information for these case studies in every subsection below. The interpretation of experimental results is also included in each subsection.

\subsection{Case study from Newtonian fluid: a second-order linear case} \label{order2linear}
\subsubsection{Problem statement}
Newtonian fluid is a type of fluid such that its viscous stress changes based on its flow. In this experiment, we aim to use simulation data to demonstrate this rule of Newtonian fluid. This would serve as a case study with second-order linear equations. 
Let $\sigma \in \mathbb{R}^{3\times 3}$ be stress tensor, 
$p \in \mathbb{R}$ pressure  and $S\in \mathbb{R}^{3\times3}$  strain rate.
The rule of Newtonian fluid is an second-order physical equation which satisfies the following condition \citep{batchelor2000introduction}: 
\begin{equation}
        \sigma = -pI+\mu S \label{newtonianFluidEqu1}
\end{equation}
Another definition of the equation for Newtonian fluid would use the velocity vector field, defined as $\nabla v$. $\nabla v$ could be expressed as a $3\times 3$ matrix. Using this definition, the equation of Newtonian fluid could also be written as: 
\begin{equation}
        \sigma = -pI+\mu (\nabla v + \nabla v^T) \label{newtonianFluidEqu2}
\end{equation}

This could also be considered as the definition of strain rate Based on this definition, we could observe that $S=\nabla v + \nabla v^T$, and it is symmetric since $S=S^T$. Since $S$ is symmetric and $I$ is an identity matrix, $\sigma$ is also symmetric. Therefore, defining an arbitrary rotation matrix $R$, this system is equipped with the property of rotation-equivariant that $R(\sigma) = R(-pI+\mu S)$. 

To quantify the stress for Newtonian fluid simulation, it would be useful to be able to predict the Newtonian fluid stress, given the simulation of pressure and velocity vector field. Based on this scenario, in this subsection, we provide a case study for the machine learning model on inputting the shear of Newtonian fluid and prediction of the stress. 
\subsubsection{Data generation and model description} \label{Expr1DGMD}
Based on Eqn.\ref{newtonianFluidEqu2}, we first generate random data to obtain $\nabla v$ and $p$. The generation of random numbers in $\nabla v$ follows a normal distribution from range $(0,1)$. Derived from generated $\nabla v$ and $p$, we could obtain $\sigma$ from Eqn.\ref{newtonianFluidEqu2}. Denote the dataset as  $D=\{x_i,y_i\}_{i=1}^{N}$. To form a proper dataset $D$ with $N$ elements for a machine learning model for Newtonian fluid, the input $x$ is set up to be a vector where $x \in \mathbb{R}^{10}$. Specifically, $x$ is composed by $p$ and flattened $S$ in Eqn.\ref{newtonianFluidEqu1}. The output $y\in \mathbb{R}^{9}$ is a vector which is the flattened result of matrix $\sigma$ derived by $p$ and $S$. The dataset $D$ would set up in the description above. To compare the difference of our method to the baseline method, we trained two models with the same hyper-parameter using different amounts of training data, ranging from $10,000, 20,000, ..., 100,000$. $85\%$ of generated data is used for training and $15\%$ of data is used for testing. A rotation set with 10,000 random rotation matrices is also generated for evaluating the property of rotation-equivariant, denoted by $\{R_i\}_{i=1}^{10000}$. 

The machine learning model we apply here is neural networks and random forests because of the ability of these two models to approximate arbitrary functions. For Neural Networks, in our implementation, the logistic activation function is used as an activation function for every neuron. The number of neurons for two layers is 512 and 4, respectively. Adam optimizer \citep{kingma2014adam} is applied as the algorithm for optimization, and the learning is set up to be $1\times 10^{-3}$. We also set the batch size to be 64. For random forests, 100 estimators are set up with mean squared error as the criterion. The maximum depth of random forests is 3 to lower the chance of overfitting. We used Sklearn for implementation \citep{pedregosa2011scikit}. The computation environment of this experiment is CPU. 

\subsubsection{Results}
\label{expr1:result}
There are two properties to evaluate, including error reduction and rotation-equivariant of RotEqNet. The effect of error reduction is evaluated for the first. A kernel predictor is trained by standard positions derived from training data. Then, the prediction algorithm is applied to both training and testing set to obtain the training and testing performances. The validation error $E$ is defined as the Mean Squared Loss using the formulation that: 

\begin{equation}
    E=\frac{\sum_{i=1}^{N} (y_i-M^{\theta}(X_i))^2}{N} \label{Error}
\end{equation}
In Eqn. \ref{Error}, $N$ is the number of data in  dataset $D$, $M$ is the trained machine learning model, $\theta$ is the derived parameter from model $M$, and $(X_i, y_i)\in D$ describes input-label pair of the dataset. This evaluation $E$ represents the expected error of model $M$ with dataset $D$. 

\begin{figure}
    \centering
    \includegraphics[scale=1.4]{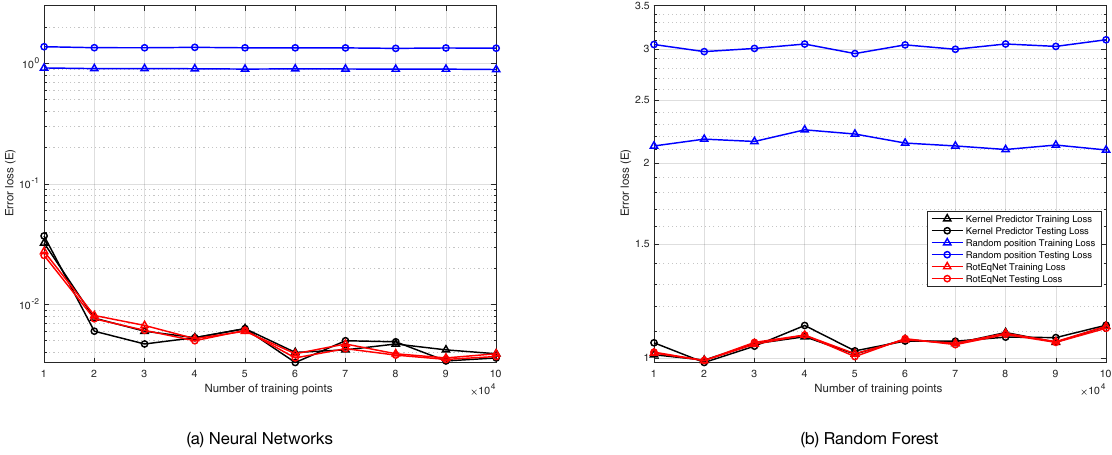}
    \caption{Error of training with baseline model with random position, RotEqNet, and kernel predictor with the standard position for (a) Neural Networks and (b) Random Forests in the case study of Newtonian Fluid. Different colors represent different experimental groups. The RotEqNet model is trained with random positions and tested with random positions (red curves). Baseline models that trained and tested on raw data are shown as blue curves. The performances of kernel predictors that trained and tested with only standard positions are also shown as black curves. Training errors are shown with lines marked with triangles, testing errors are shown with lines marked with circles. }
    \label{fig:expr1}
\end{figure}

\begin{table} 
\centering
\begin{tabular}{ |p{2.5cm}|p{4.0cm}|p{4.0cm}| }
 \hline
 Kernel predictor& Training Error Reduction &Testing Error Reduction\\
 \hline
 Neural Networks   & 99.56\%    &99.60\%\\
  \hline
 Random Forests &   99.56\%  & 99.72\% \\

 \hline
\end{tabular}
\label{tab1:expr1}
\caption{Evaluation of error reduction for RotEqNet with different kernel predictor.}
\end{table}

Fig. \ref{fig:expr1}(a) shows the error reduction property of RotEqNet. This plot consists of three groups of experimental groups. 
The first experiment group focuses on the accuracy of the baseline model, a single feed-forward Neural Network, on raw data with random rotated positions. As shown in Fig. \ref{fig:expr1}(a) with blue curves, triangle curve represents training error and circle curve represents testing error. The second experiment group is RotEqNet with Neural Network as the kernel predictor. As shown in Fig. \ref{fig:expr1}(a) in red curves, triangle curve represents training error, and the circle curve represents testing error. For 100,000 training samples, the testing error of RotEqNet is 0.0037, and the testing error of the baseline method is 1.333. We could observe a huge error reduction, 99.56\% in training, and 99.60\% in testing, for RotEqNet compared to the error of using the baseline model.
For the last experiment group, as performances marked as black curves in the figure, it reports the performance of kernel predictor with standard position only. This experiment would explain why RotEqNet would improve performance since training with standard positions would be an easier task compared to raw data. 

Further, Fig. \ref{fig:expr1}(b) shows the error reduction effect of RotEqNet using Random Forest as a kernel predictor. Similarly, as shown in Fig. \ref{fig:expr1}(b) with blue curves, it represents the performance of the baseline method (Random Forests). The second experiment group is RotEqNet with Random Forests as the kernel predictor. As shown in Fig. \ref{fig:expr1}(b) in red curves, triangle curve represents training error and the circle curve represents testing error. We could observe a huge error reduction, 99.56\% in training and 99.72\% in testing, for RotEqNet compared to the error of using only the Random Forest predictor.
For the last experiment group, as performances marked as black curves in the figure, trains the kernel predictor with standard position only. As stated before, this could also serve as a reason for the error reduction effect for RotEqNet on random forests. 

According to the reported results, RotEqNet has a good generalization result without overfitting. For cases training with raw data for baseline models, the testing error is typically higher compared to training error. For example, the difference is training and testing errors are 0.44 for Neural Networks, $1.01$ for Random forests when $N=100,000$. This represents that for both Neural Networks and Random Forests would be easy to overfit this task on Newtonian Fluid. By contrast, RotEqNet would help to reduce this difference in training and testing error. As we could observe from the training and testing error of RotEqNet, the errors are much lower. When $N=100,000$, there are only $0.0002$ for Neural Networks and $0.0078$ for Random Forests. In the case of linear second-order equations, the application of RotEqNet would result in better-generalized results in learning. 

Another important property to evaluate for RotEqNet is rotation-equivariant. The experiment is designed on the definition of rotation-equivariant mentioned in Eqn. \ref{def:rotEq}. First, we pick a randomly generated data $(X_0, y_0)$. Then we apply the rotation set $\{R_i\}_{i=1}^{10000}$ with 10,000 random rotation matrices to $(X_0, y_0)$. To fully investigate the property of rotation-equivariant, we apply an error evaluation method $E_D$ here to evaluate the error compared to real data, which is defined as: 
\begin{equation}
       E_D=\frac{\sum_{i=1}^{N} [(M^{\theta}(R_i(X_0))-R_i(y_0)]^2}{N}
    \label{ED}
\end{equation}

\begin{table} 
\centering
\begin{tabular}{ |p{1.3cm}|p{2.4cm}|p{2.4cm}|p{2.4cm}|p{2.4cm}| }
 \hline
 Model& Baseline (NN)& RotEqNet (NN)&Baseline (RF)&RotEqNet (RF)\\
 \hline
  $E_D$ & 0.6362 & \textbf{0.0013} &3.1334 &  1.5513\\

 \hline
\end{tabular}
\caption{Evaluation of Rotation-equivariant property between baseline model and RotEqNet.}
\label{tableExpr1}
\end{table}

This error evaluation method ($E_D$) focuses more on the model's error on real data for all rotations. As shown in Tab. \ref{tableExpr1}, for both baseline methods, using neural networks and random forests, there are large errors for $E_M$ and $E_D$. The baseline methods have no theoretical guarantee that it has the property of rotation-equivariant. However, there is an error reduction for both machine learning models when applying with RotEqNet's architecture. Especially for RotEqNet with Neural Networks as kernel predictor, we could observe that $E_D=0.0013$ with $99.8\%$ of error reduction. This could guarantee the rotation-equivariant property of RotEqNet.

\subsection {Case study from large eddy simulation: a second-order nonlinear case} \label{order2_nonlinear}
\subsubsection{Problem statement}
In this case, we consider the subgrid model of large eddy simulation (LES) of turbulent flow by Kosovic \cite{kosovic1997subgrid}. In this case study, as formulated previously in \citep{pitsch2006large,matai2018flow}, we hope to obtain a learned model by simulation data from LES. This would serve as a case study with second-order non-linear equations. LES is defined as: 
\begin{equation} 
       \tau_{ij} = -(C_s \Delta)^2\left\{ 2\sqrt{2S_{mn}S_{mn}}S_{ij} + C_1\left( S_{ik}S_{kj}-\frac{1}{3}S_{mn}S_{mn}\delta_{ij}\right)+C_2\left( S_{ik}\Omega_{kj} - \Omega_{ik}S_{kj} \right)\right\} \label{LESEqu}
\end{equation}
 Here $\tau_{ij}$ is subgrid stress, which is a symmetric traceless 2nd order tensor. $S_{ij}$ and $\Omega_{ij}$ are symmetric and anti-symmetric parts of velocity gradient tensor $G_{ij}$, where $\mathrm{Tr}(G) = 0$. Further, $C_s$, $\Delta$, $C_1$, $C_2$ are all constants. The configuration of constants above are reported in the next subsection. 

In order to qualify the subgrid stress for LES, this study aims to predict the subgrid stress, given the simulation of velocity gradient tensor. This case study for the machine learning model on inputting the velocity gradient tensor. 
\subsubsection{Data generation and model description}
Based on Eqn. \ref{LESEqu}, we first generate random data to obtain a simulated velocity gradient tensor $G_{ij}$. The generation of random numbers follows a normal distribution from range $(0,1)$, and $G_{ij}$ is obtained from a random matrix $G_{raw}$ by subtracting $\frac{1}{3}\mathrm{Tr}(G_{raw})$ from diagonal position. This would keep $\mathrm{Tr}(G)=0$. $S_{ij}$ and $\Omega_{ij}$ could be obtained from $G_{ij}$ by getting its symmetric and anti-symmetric parts. For the setup of constants, $C_s=0.4,\Delta=0.4,C_1=C_2=1.0$. $tau_{ij}$ is computed from the above setting with Eqn. \ref{order2_nonlinear}. Denote the dataset as, $D=\{x_i,y_i\}_{i=1}^{N}$. To form a proper dataset $D$ with $N$ elements for a machine learning model for Newtonian fluid, the input $x$ is set up to be a vector where $x \in \mathbb{R}^{9}$. Specifically, $x$ is composed by flattened $G_{ij}$. The output $y\in \mathbb{R}^{9}$ is a vector, which is the flattened result of matrix $\tau$ derived by $G$ and other constants. To compare the difference of our method to the baseline method, we trained two models with the same hyper-parameter using different amounts of training data, ranging from $10,000, 20,000, ..., 100,000$. $85\%$ of generated data is used for training, and $15\%$ of data is used for testing. A rotation set with 10,000 random rotation matrices is also generated for evaluating the property of rotation-equivariant, denoted by $\{R_i\}_{i=1}^{10000}$. The model setup is the same compared to Sec. \ref{Expr1DGMD}. 

\subsubsection{Results}
The effect of error reduction is evaluated for the first. The validation error $E$ is defined as the Mean Squared Loss using the formulation in Eqn. \ref{Error}. This evaluation $E$ represents the expected error of model $M$ with dataset $D$.

\begin{figure}
    \centering
    \includegraphics[scale=1.4]{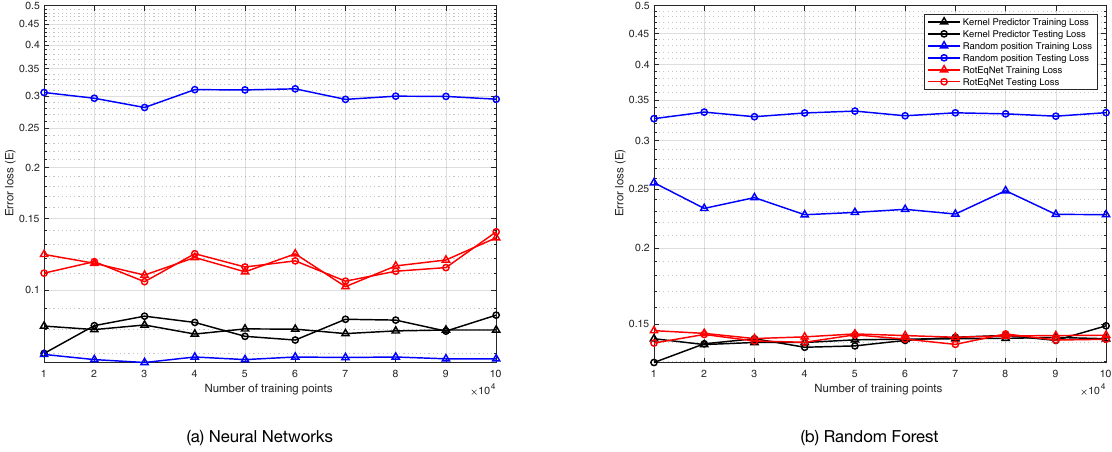}
    \caption{Error of training with baseline model with random position, RotEqNet, and kernel predictor with the standard position for (a) Neural Networks and (b) Random Forests in the case study of large eddy simulation. Different colors represent different experimental groups. The RotEqNet model is trained with random positions and tested with random positions (red curves). Baseline models that trained and tested on raw data are shown as blue curves. The performances of kernel predictors that trained and tested with only standard positions are also shown as black curves. Training errors are shown with lines marked with triangles, testing errors are shown with lines marked with circles. }
    \label{fig:expr2}
\end{figure}

Fig. \ref{fig:expr2}(a) shows the error reduction effect of RotEqNet with Neural Network as a kernel predictor for second-order nonlinear cases with three groups of experimental groups. 
The first experiment group focuses on the accuracy of the baseline method on raw data with random rotated positions. As shown in Fig. \ref{fig:expr2} with blue curves, triangle curve represents training error and the circle curve represents testing error. The second experiment group is RotEqNet, with Neural Network as a kernel predictor. As shown in Fig. \ref{fig:expr2}(a) in red curves. For 100,000 training samples, the testing error of RotEqNet is 0.1391, and the testing error of the baseline method is 0.2946, with 52.77\% of error reduction.
The performances of the last experiment group are marked as black curves in the figure, with only standard position trained for kernel predictor.

Based on the experimental results, firstly, RotEqNet could reach a better learning performance compared to simply applying Neural Networks (baseline method). Training with standard positions could lower the training difficulty, and therefore RotEqNet could obtain better performance. 
Further, Fig. \ref{fig:expr2}(b) shows the error reduction effect of RotEqNet using Random Forest as a kernel predictor. The general performance of using Random Forests as a kernel predictor is relatively worse compared to using Neural Networks as a kernel predictor. In Fig. \ref{fig:expr2}(b), blue curves represent the performance of training with raw data by Random Forests (baseline method); red curves represent the performance of RotEqNet; black curves represent the performance of kernel predictor trained with standard positions. We could observe an error reduction for 36.63\% in training and 57.58\% in testing for RotEqNet with Random Forests. 

Moreover, RotEqNet has a good generalization result without overfitting. Applying raw data in learning directly on baseline models, the testing error is much higher compared to the training error. For example, the difference is training and testing errors are $0.0068$ for Neural Networks, $0.1068$ for Random forests when $N=100,000$. It is also observable in Fig. \ref{fig:expr2}(a) that the training error of the baseline model with raw data is the lowest, while the testing error of the baseline model is the highest. In this case study, Neural Networks are worse for the sake of overfitting compared to Random Forests. By contrast, introducing the architecture RotEqNet would help to reduce this difference in training and testing error. As we could observe from the training and testing error of RotEqNet, the errors are much lower. When $N=100,000$, there are only $0.0046$ for Neural Networks and $0.0022$ for Random Forests. In this case study of LES, the application of RotEqNet would result in better-generalized results in learning. 

\begin{table} 
\centering
\begin{tabular}{ |p{2.5cm}|p{4.0cm}|p{4.0cm}| }
 \hline
 Kernel predictor& Training Error Reduction &Testing Error Reduction\\
 \hline
 Neural Networks   & -98.44\%    &52.77\%\\
  \hline
 Random Forests &   36.63\%  & 57.58\% \\

 \hline
\end{tabular}
\caption{Evaluation of error reduction for RotEqNet with different kernel predictor.}
\label{tab1:expr2}
\end{table}

To evaluate the rotation-equivariant property of RotEqNet for second-order nonlinear cases, our experimental process is close to the one stated in Sec. \ref{expr1:result}. First, we pick a randomly generated data $(X_0, y_0)$. Then we apply the rotation set $\{R_i\}_{i=1}^{10000}$ with 10,000 random rotation matrices to $(X_0, y_0)$. 
This error evaluation method ($E_D$), as defined in Eqn. \ref{ED}, focuses more on the model's error on real data for all rotations. As shown in Tab. \ref{tableExpr2}, for both baseline methods, using neural networks and random forests, there are large error for $E_D$. The baseline methods have no theoretical way to guarantee that it has the property of rotation-equivariant. However, there is an error reduction for both machine learning models when applying with RotEqNet's architecture. Especially for RotEqNet with Neural Networks as kernel predictor, it is observable that $E_D=0.0025$ with $73.55\%$ error reduction. This could guarantee the rotation-equivariant property of RotEqNet for nonlinear second-order cases.
\begin{table} 
\centering
\begin{tabular}{ |p{1.3cm}|p{2.4cm}|p{2.4cm}|p{2.4cm}|p{2.4cm}| }
 \hline
 Model& Baseline (NN)&RotEqNet (NN)&Baseline (RF)& RotEqNet (RF)\\
 \hline
  $E_D$ & 0.0945 & \textbf{0.0025} &0.1912 &  0.0084\\

 \hline
\end{tabular}
\caption{Evaluation of Rotation-equivariant property between baseline model and RotEqNet.}
\label{tableExpr2}
\end{table}

\subsection{Case study from testing Newtonian Fluid equation: a third-order case} \label{order3}
\subsubsection{Problem statement}
In this section, we study the performance of RotEqNet for tensor with odd order. In this case, we specifically set a third-order test equation. We used a test equation here revised from Newtonian fluid equation from Eqn. \ref{newtonianFluidEqu2}. We name this testing equation as 'testing Newtonian fluid equation' for simplicity. The testing equation revised from Newtonian fluid equation can be described as: 
\begin{equation}
       \sigma = -pI+\mu (\nabla v + \nabla v^T) \label{newtonianFluidEqu3},
\end{equation}
where $\sigma \in \mathbb{R}^{3\times 3\times 3}$ is testing stress, $p \in \mathbb{R}$ is testing pressure, and $v\in \mathbb{R}^{3\times 3\times 3}$ is testing velocity field. $I\in \mathbb{R}^{3\times 3\times 3}$ is the identity third-order tensor. 

Based on this testing equation, we could observe that $(\nabla v + \nabla v^T)^T=\nabla v + \nabla v^T$. Since $\nabla v + \nabla v^T$ is symmetric, and $I$ is symmetric, we have $\sigma$ is also symmetric. Therefore, defining an arbitrary rotation matrix $R$, this system is equipped with the property of rotation-equivariant that $R(\sigma) = R(-pI+\mu (\nabla v + \nabla v^T))$. 

In order to qualify stress for testing the Newtonian fluid equation, this study aims to predict the stress, given the simulation of pressure and velocity gradient tensor. This case study for the machine learning model on inputting the pressure and velocity gradient tensor. 
\subsubsection{Data generation and model description}
Based on Eqn. \ref{newtonianFluidEqu3}, we first generate random data to obtain $\nabla v$ and $p$. The generation of random numbers in $\nabla v$ follows a normal distribution from range $(0,1)$. $\sigma$ could be obtained using the Eqn.\ref{newtonianFluidEqu3}, derived from generated $\nabla v$ and $p$. Denote the dataset as , $D=\{x_i,y_i\}_{i=1}^{N}$. To form a proper dataset $D$ with $N$ elements for a machine learning model for Newtonian fluid, the input $x$ is set up to be a vector where $x \in \mathbb{R}^{28}$. Specifically, $x$ is composed by $p$ and flattened $(\nabla v + \nabla v^T)$ in Eqn.\ref{newtonianFluidEqu3}. The output $y\in \mathbb{R}^{27}$ is a vector which is the flattened result of matrix $\sigma$. The dataset $D$ would set up in the description above. To compare the difference of our method to the baseline method, we trained two models with the same hyper-parameter using different amounts of training data, ranging from $10,000, 20,000, ..., 100,000$. $85\%$ of generated data is used for training and $15\%$ of data is used for testing. A rotation set with 10,000 random rotation matrices is also generated for evaluating the property of rotation-equivariant, denoted by $\{R_i\}_{i=1}^{10000}$. The model setup is the same as Sec. \ref{Expr1DGMD}.

\subsubsection{Results}
Fig. \ref{fig:expr3}(a) shows the error reduction effect of RotEqNet with Neural Network as a kernel predictor for third-order cases with three groups of experimental groups. 
The first experiment group focuses on the accuracy of the baseline model (Neural Network) on raw data with random rotated positions as shown in Fig. \ref{fig:expr3}(a) with blue curves. The second experiment group is RotEqNet, with Neural Network as kernel predictor as shown in Fig. \ref{fig:expr3}(a) in red curves. For 100,000 training samples, the testing error of RotEqNet is 1.8759 and the testing error of baseline method is 2.2232 with 15.62\% of error reduction.
The performances of the last experiment group are marked as black curves in the figure, with only standard position trained for kernel predictor.

Based on the experimental results, for the third-order testing equation, RotEqNet could reach a better learning performance compared to the baseline method. Training with RotEqNet could lower the training difficulty, and therefore RotEqNet could obtain better performance. Moreover, RotEqNet has good generalization capability without overfitting. As shown in the blue curves of Fig. \ref{fig:expr3}, if we apply raw data in learning directly on baseline models, the testing error is much higher compared to the training error. In this case study, introducing the architecture RotEqNet would help to reduce this difference in training and testing error. As we could observe from the training and testing error of RotEqNet, the errors are much lower. When $N=100,000$, there are only $0.0046$ for Neural Networks and $0.0022$ for Random Forests. In this case study of testing the Newtonian fluid equation, the application of RotEqNet would result in better-generalized results in learning. 

Further, Fig. \ref{fig:expr3}(b) shows the error reduction effect of RotEqNet using Random Forest as a kernel predictor. The general performance of using Random Forests as a kernel predictor is relatively worse compared to using Neural Networks as a kernel predictor. In Fig. \ref{fig:expr2}(b), blue curves represent the performance of training with raw data by Random Forests (baseline method); red curves represent the performance of RotEqNet; black curves represent the performance of Random Forest trained with standard positions. For the first point, we could observe an error reduction for 0.90\% in training and 6.84\% in testing for RotEqNet with Random Forests. As another point, RotEqNet is also obtaining a better-learned model for the model's capability in generalization. The testing error of the baseline method is observably higher than testing error, while the training and testing performance of RotEqNet is approximately the same. As suggested in Fig. \ref{fig:expr3}(a), in second-order nonlinear cases, RotEqNet could reach a generalized learning result with remarkably lower error compared to baseline methods. 

\begin{figure}
    \centering
    \includegraphics[scale=1.4]{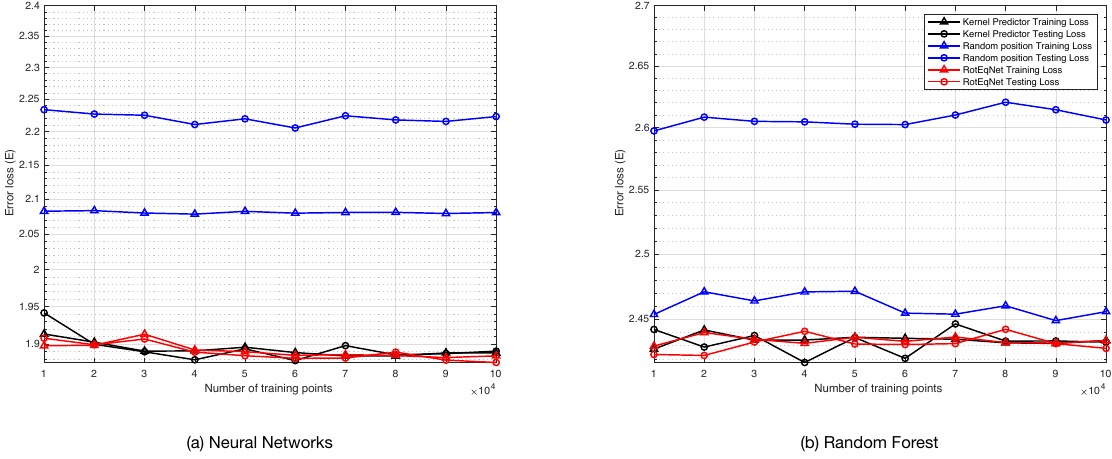}
    \caption{Error of training with baseline model with random position, RotEqNet, and kernel predictor with the standard position for (a) Neural Networks and (b) Random Forests in the case study of testing Newtonian Fluid equation. Different colors represent different experimental groups. The RotEqNet model is trained with random positions and tested with random positions (red curves). Baseline models that trained and tested on raw data are shown as blue curves. The performances of kernel predictors that trained and tested with only standard positions are also shown as black curves. Training errors are shown with lines marked with triangles, testing errors are shown with lines marked with circles. }
    \label{fig:expr3}
\end{figure}

\begin{table} 
\centering
\begin{tabular}{ |p{2.5cm}|p{4.0cm}|p{4.0cm}| }
 \hline
 Kernel predictor& Training Error Reduction &Testing Error Reduction\\
 \hline
 Neural Networks   & 9.42\%    &15.62\%\\
  \hline
 Random Forests &   0.90\%  & 6.84\% \\

 \hline
\end{tabular}
\caption{Evaluation of error reduction for RotEqNet with different kernel predictor.}
\label{tab1:expr3}
\end{table}

\begin{table} 
\centering
\begin{tabular}{ |p{1.3cm}|p{2.4cm}|p{2.4cm}|p{2.4cm}|p{2.4cm}| }
 \hline
 Model& Baseline (NN)&RotEqNet (NN)&Baseline (RF)& RotEqNet (RF)\\
 \hline
  $E_D$ & 2.8454 &\textbf{2.6992} & 3.0788 &  3.1068\\

 \hline
\end{tabular}
\caption{Evaluation of Rotation-equivariant property between baseline model and RotEqNet.}
\label{tableExpr3}
\end{table}

To evaluate the rotation-equivariant property of RotEqNet for this third-order case, we designed an experimental process that is close to the one stated in Sec. \ref{expr1:result}. As shown in Tab. \ref{tableExpr3}, for baseline method using neural networks, the error is relatively large for $E_D$ compared to RotEqNet. In our experiment, we reached an error reduction of $0.1462$. We would further discuss this result in Section \ref{discussion}.

\subsection{Case study from Electrostriction: a fourth-order case}
\subsubsection{Problem statement}
This case study focuses on a linear relationship of fourth-order tensor. Nye \cite{nye1985physical} has introduced a fourth-order tensor in modeling elastic compliances and stiffnesses, which has been investigated using machine learning methods \citep{yang2019predicting,liu2019deep}. Generally, in the study of the properties of a crystalline and anisotropic elastic medium, a fourth-order tensor coefficient will typically be applied to model the relationship between two symmetric second-order tensors \cite{walpole1984fourth}. In this case, we study Electrostriction, a property causing all electrical non-conductors to change their shape under the application of an electric field. The relationship is described as:  
\begin{equation} 
       T_{ij}=V_{ijkl}V_{kl} \label{4th-order}
\end{equation}
 Here $T_{ij} \in \mathbb{R}^{3\times 3}$ is a symmetric traceless second-order strain tensor. $S_{kl} \in \mathbb{R}^{3\times 3}, S_{kl}=S_k S_l$ where $V_k$ and $V_l$ are first-order electric polarization density. Note here $V_{kl}$ is symmetric. $V_{ijkl}\in \mathbb{R}^{3\times 3\times 3 \times 3}$ is the electrostriction coefficient. 

Based on the formulation above, this system is symmetric. Since $S_{ij}$ is symmetric, $T_{ij}^T=(V_{ijkl} S_{ij})^T=S_{kl} V_{klij}=T_{ij}$. This could guarantee that $T_{ij}$ is also symmetric. Due to the face that the system is symmetric, applying a random rotation matrix $R$, $R(T)=R(VS)$.

In order to qualify strain for study on Electrostriction, we aim to predict the strain, given the simulation of electrostriction coefficient and electric polarization density. 
\subsubsection{Data generation and model description}
Based on Eqn. \ref{4th-order}, we first generate random data to obtain simulated electrostriction coefficient tensor $V_{ijkl}$ and electric polarization density tensor $S_{ij}$. The generation of random numbers follows a normal distribution. $T_{ij}$ is computed from above setting using $V_{ijkl}$ and $S_{ij}$. Denote the dataset as, $D=\{x_i,y_i\}_{i=1}^{N}$. To form a proper dataset $D$ with $N$ elements for a machine learning model for the study on Electrostriction, the input $x$ is set up to be a vector where $x \in \mathbb{R}^{90}$. Specifically, $T_{ij}$ is composed by flattened $V_{ijkl}$ and $S_{ij}$. The output $y\in \mathbb{R}^{9}$ is a vector, which is the flattened result of second-order tensor $T$. To compare the difference of our method to the baseline method, we trained two models with the same hyper-parameter using different amounts of training data, ranging from $10,000, 20,000, ..., 100,000$. $85\%$ of generated data is used for training, and $15\%$ of data is used for testing. A rotation set with 10,000 random rotation matrices is also generated for evaluating the property of rotation-equivariant, denoted by $\{R_i\}_{i=1}^{10000}$. The model setup is the same compared to Sec. \ref{Expr1DGMD}. We use Numpy to generate this simulated dataset by generating a random symmetric fourth-order tensor $V$, and second-order tensor $S$. $T$ is computed from generated $V$ and $S$ by Eqn. \ref{4th-order}. 
\begin{figure}
    \centering
    \includegraphics[scale=1.4]{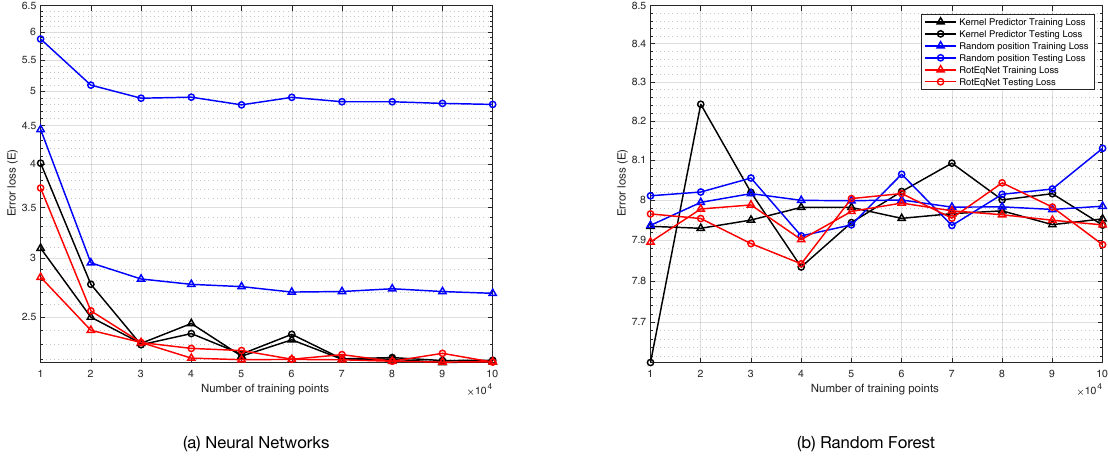}
    \caption{Error of training with baseline model with random position, RotEqNet, and kernel predictor with the standard position for (a) Neural Networks and (b) Random Forests in the case study of Electrostriction. Different colors represent different experimental groups. The RotEqNet model is trained with random positions and tested with random positions (red curves). Baseline models that trained and tested on raw data are shown as blue curves. The performances of kernel predictors that trained and tested with only standard positions are also shown as black curves. Training errors are shown with lines marked with triangles, testing errors are shown with lines marked with circles. }
    \label{fig:expr4}
\end{figure}

\subsubsection{Results}
\label{sec4:order4:result}
The effect of error reduction is evaluated for the first. The validation error $E$ is defined as the Mean Squared Loss using the formulation in Eqn. \ref{Error}. This evaluation $E$ represents the expected error of model $M$ with dataset $D$. Fig. \ref{fig:expr4} shows the performance of Neural Networks and Random Forests as kernel predictor separately. It is observable that in high-order cases, Neural Networks have huge superiority to Random Forests. Details will be demonstrated in the following paragraphs.

We are starting with Neural Networks, Fig. \ref{fig:expr4}(a) shows the error reduction effect of RotEqNet with Neural Network as a kernel predictor. 
As shown in blue curves, the first experiment group focuses on the accuracy of the baseline model on raw data with random rotated positions. The second experiment group is RotEqNet marked with red curves. As shown in black curves, it shows the performance of the kernel predictor trained by standard position. For 10,000 training samples, the testing error of RotEqNet is 4.0106 and the testing error of baseline model is 8.6458 with 53.61\% of error reduction. The testing performance of the kernel predictor is only evaluated on the testing set with only standard positions. It will be helpful to explain the reason for the improved performance of RotEqNet. 

To interpret the experimental results, firstly, RotEqNet could reach a better learning performance compared to simply applying Neural Networks (baseline method). A dataset with only standard positions has lower training difficulty compared to random positions. This claim is supported by black curves in Fig. \ref{fig:expr4}(a), the performance of the kernel predictor is much better than the baseline model. RotEqNet could obtain better performance for utilizing rotation symmetry as a prior, and training kernel predictor with only standard positions. Moreover, RotEqNet has a good generalization result without clear overfitting. The training error and testing error of RotEqNet is considerably close to each other, and sometimes, the testing error of RotEqNet is even slightly better than its training error. By contrast, applying raw data in learning directly on $M_{baseline}$ would result in an overfitted model. The testing error is much higher compared to the training error. To demonstrate the improved learning result in generalization, for example, when $N=100,000$, the difference between training and testing errors for RotEqNet is only $0.0024$ while the difference of the baseline method is $2.1118$. As a quick conclusion, for Neural Networks as a kernel predictor, the application of RotEqNet would be better compared to the baseline method. 

\begin{table} 
\centering
\begin{tabular}{ |p{2.5cm}|p{4.0cm}|p{4.0cm}| }
 \hline
 Kernel predictor& Training Error Reduction &Testing Error Reduction\\
 \hline
 Neural Networks   & 18.93\%    &	54.63\%\\
  \hline
 Random Forests &   0.58\%  & 2.96\% \\

 \hline
\end{tabular}
\caption{Evaluation of error reduction for RotEqNet with different kernel predictor.}
\label{tab1:expr4}
\end{table}

Further, Fig. \ref{fig:expr4}(b) shows the error reduction effect of RotEqNet using Random Forest as a kernel predictor. 
At first glance, we could find that the curves for Random Forests are quite messy without certain patterns like Fig. \ref{fig:expr4}(a). The general performance of using Random Forests as a kernel predictor is worse in both aspects of performance and generalization. In Tab. \ref{tab1:expr4}, we could observe a training error reduction for 0.58\% and testing error reduction of 2.96\%. Even if we could still see the general error of RotEqNet seems to be slightly lower than the baseline method. This result is not comparable to the error reduction performance with setting Neural Networks as a kernel predictor. 
As another point, selecting Random Forests as a kernel predictor fails to extract learning rules with the standard position. As we could observe the black curves in Fig. \ref{fig:expr4}(b) is not showing an improved performance as good as using Neural Networks. 
Finally, the learned model of RotEqNet is also not getting a model with better generalization capability. There is no significant reduction of overfitting error compared to the baseline method. 

\begin{table} 
\centering
\begin{tabular}{ |p{1.3cm}|p{2.4cm}|p{2.4cm}|p{2.4cm}|p{2.4cm}| }
 \hline
 Model& Baseline (NN)&RotEqNet (NN)&Baseline (RF)& RotEqNet (RF)\\
 \hline
  $E_D$ & 3.9290 & \textbf{2.7960} &4.8976 &  4.8740\\

 \hline
\end{tabular}
\caption{Evaluation of Rotation-equivariant property between baseline model and RotEqNet.}
\label{tableExpr4}
\end{table}
To evaluate the rotation-equivariant property of RotEqNet for this fourth-order case, we designed an experimental process as stated in Sec. \ref{expr1:result}. 
The error evaluation measurement ($E_D$), as defined in Eqn. \ref{ED}, focuses more on the model's error on real data for all rotations. As shown in Tab. \ref{tableExpr4}, when using neural networks, baseline method has large error for $E_D$. RotEqNet helps in keeping the rotation-equivariant property as observing error reduction in $E_D$ for $28.86\%$. Considering the case using Random Forests as a kernel predictor, as shown in the previous paragraph, because of the reason that Random Forests are relatively bad in learning fourth-order data, the performance of $E_M$ is still affected, which results in a large error in the prediction of RotEqNet with Random Forest.

\label{discussion}
The large error reduction observed in case studies raised new opportunities in solving the problem of the physical system with rotation symmetry. Most physical systems have the property of rotation symmetry, and currently, there exist few works that could provide a theoretical guarantee to this property for machine learning methods. A key point in this problem is to design a properly defined algorithm to obtain rotation invariant for high-order tensors. This paper has shown RotEqNet with theoretical and experimental results aiming to solve the problem of rotation symmetry.

We first define a standard position as rotation invariant, which is compatible for high-order tensors. It allows us to extract the rotation invariant of high-order tensors using a contraction, diagonalization, and QR factorization. The theoretical guarantee is shown in Thm. \ref{thm:2}, and the algorithm is shown in Alg. \ref{rotexAlgo}. RotEqNet is built on Alg. \ref{rotexAlgo} with a kernel predictor which only deals with standard positions (rotation invariants). By setting kernel predictor with Neural Networks and Random Forests, these two methods are compared with baseline methods in four different case studies focusing on second-order linear, second-order nonlinear, third-order linear, and fourth-order linear cases. There are three important points to address from the observation of case studies. 

First, the definition of the standard position is successful. The definition of the standard position is not unique. We aim to define a proper version of the standard position to simplify the learning task by removing the effect of rotation symmetry. In our case, the standard position satisfies the definition of rotation-invariants, which selects a representative point from the orbit of an element via diagonalization (or QR factorization). The experimental results are compatible with this definition of the standard position. We could observe in most of the cases, training kernel predictors with only rotation invariants could reach the lowest error. The reduced error means that the rotation invariant in our definition could lower the difficulty of this learning task as we previously discussed the reason in Sec. \ref{reason}. 

Second, RotEqNet is equipped with the property of Rotation-equivariant. As we could observe from the results of case studies, the rotation error $E_M$ is typically low compared to baseline methods. 
The perseverance of the property of Rotation-equivariant shows the successful design of RotEqNet and the correctness of Thm. \ref{thm:2}. Operating with Alg. \ref{rotexAlgo}, the property of Rotation-equivariant of RotEqNet could be held if and only if Thm. \ref{thm:2} is correct. 
Further, this fact would cause an error reduction for RotEqNet. As stated in the previous paragraph, training with rotation invariants will result in a lower error. Under this situation, adding with the property of Rotation-equivariant, this would cause RotEqNet could process this system with any rotation. 

The two reasons above are the main reasons that are causing the error reduction for RotEqNet. There is also one point to mention is the selection of kernel predictor. The model selection of kernel predictor will affect the learning results significantly since the kernel predictor is essential in learning the physical system without the effect of rotation symmetry. Neural Networks is the best model in the design of the data-driven method for physical systems because of its flexibility to approximate arbitrary functions. We only reported the performance of Neural Networks and Random Forests as previous work by Ling \cite{ling2016machine}. As described in Sec. \ref{sec4:order4:result}, the performance of Random Forests is limited compared to Neural Networks. Also, as a general trend in previous experiments, Neural Networks are usually reaching better performance compared to Random Forests. As a quick conclusion, we believe the application of Neural Networks as a kernel predictor has a series of advantages than other machine learning models. 

We wish to further discuss about another error evaluation method of rotation-equivariance property that we do not mention in case studies. Consider a type of error evaluation, evaluating rotation error of model itself, the error $E_M$ is defined as: 
\begin{equation}
       E_M=\frac{\sum_{i=1}^{N} [(M^{\theta}(R_i(X_0))-R_i( M^{\theta}(X_0))]^2}{N}
    \label{EM}
\end{equation}
The evaluation of this error is actually trivial since we have already rigorously provided a proof in Theorem \ref{thm:2} showing the rotation-equivariance property of RotEqNet. We applied this evaluation in first two case studies, and the estimated error is around $10^{-16}$ for all these cases.

For future work, there are three directions to this paper: a better definition of standard position, application to other groups, and generalization to non-symmetric systems. For the first direction, for the current definition, the rotation invariant of odd-order tensors is not reaching equivalent performance as even-order tensors. It would be a good work for revising the definition of standard position for odd tensors. Second, besides rotation symmetries, there are also physical systems with other group-equivariant properties such as scaling and transaction. This work could provide a method in solving problems with other groups, but the detailed design of an algorithm should differ from case to case. Third, current work could only deal with the symmetric system. However, for a general case, if the system is not symmetric, there are certain methods to use RotEqNet in a symmetric system for solving a non-symmetric system. A good trick to consider, for example, is to deal with $PP^T$, where $P$ is a matrix. This is a great intuition to extend our current work into non-symmetric physical systems. 

\section*{Acknowledgments}
 G. Lin would like to acknowledge the support from National Science Foundation (DMS-1555072 and DMS-1736364).

\section{Appendix}
\renewcommand{\theequation}{A-\arabic{equation}}

\label{appendix:tensor}
\subsection{Lemma 2.1}
\begin{proof}
We will use column vector convention to represent vectors in $V$. 
Let $v_1$ and $v_2$ be vectors in $V$. Then
\begin{equation}
    M(v_1\otimes v_2) = M(v_1)\times M(v_2)^t
\end{equation}
Then, 
\begin{align}
    M(R(v_1\otimes v_2)) &=  M(R(v_1) \otimes R(v_2))\\
                        & = M(R(v_1)) \times M(R(v_2))^t \\
                        & = M(R)\times M(v_1)\times M(v_2)^t\times M(R)^t \\
                        & = M(R)\times M(v_1\otimes v_2) \times M(R)^t
\end{align}
Therefore, 
\begin{equation}
    M(R(T)) = M(R)\times M(T) \times M(R)^t
\end{equation}  
\end{proof}
\subsection{Lemma 2.2}
\begin{proof}
We will use column vector convention to represent vectors in $V$. 
Let $v_1$ be vector in $V$. 
Then
\begin{equation}
    M(R(v_1))  = M(R)\times M(v_1)
\end{equation}
Therefore, 
\begin{equation}
    M(R(T)) = M(R)\times M(T)
\end{equation}  
\end{proof}
\subsection{Lemma 2.3}

\begin{proof}
Since both $C(a, b)$ and $g$ are linear, we may assume that $T$ is of 
the form $v_{i1}\otimes\cdots\otimes v_{in}$. 
\begin{align}
    C(a,b)(g(T)) & = C(a,b)(g(v_{i1})\otimes\cdots\otimes g(v_{in})) \\
        & = \langle g(v_{ia}), g(v_{ib})\rangle g(v_{i1})\otimes \cdots
        \check{g(v_{ia})}\cdots\check{g(v_{ib})}\cdots\otimes g(v_{in}) 
\end{align}
Since $g$ is a rotation, it preserves the inner product i.e.
\begin{equation}
    \langle g(v_{ia}), g(v_{ib})\rangle = \langle v_{ia}, v_{ib} \rangle
\end{equation}
So 
\begin{align}
    C(a,b)(g(T)) & = C(a,b)(g(v_{i1})\otimes\cdots\otimes g(v_{in})) \\
        & = \langle v_{ia}, v_{ib}\rangle g(v_{i1})\otimes \cdots
        \check{g(v_{ia})}\cdots\check{g(v_{ib})}\cdots\otimes g(v_{in}) \\
        & = g(C(a, b)(T))
\end{align}
\end{proof}

\subsection{Proof of Theorem 1}
\begin{proof}
\label{proof:thm2}

Since the position standardization algorithm defines standard position differently for even and odd orders. We show our proof on even and odd cases separately. 

Suppose $T$ has even order. 

Let $\mathcal{C}$ be the sequence of contraction
along the first two axes such that $C(T) = T^2$, where $T^2$ is
a second-order tensor as described in the algorithm. 

Given arbitrary even high order tensor $T$, we could perform contraction to a second order tensor $T^2$ via first two indices: 
\begin{equation}
    T^2 = \mathcal{C}(T^n)
\end{equation}
For $T^2$, using Lemma \ref{lemma2.1}, there exists a
rotation $R$ such that: 
\begin{equation}
    T^2=R(T^{2}_s),
    \label{T_2}
\end{equation}
where $R(T^{2}_s) = R T^{2}_s R^t$. $T_2$ is diagonalizable because it is symmetric. Since $R$ is represented
by a orthonormal matrix, therefore $R^t = R^{-1}$.

Based on Lemma \ref{lemma2.3}, we know rotation commutes with contraction. Therefore, based on $S$ the standard position $S(T)$ is defined as
\begin{equation}
    S(T)=R^{-1}(T)
\end{equation}
Consider a rotation operation $P$ in its matrix form. When we act $P$ on $T$ we obtain a new tensor $P(T)$. For this new tensor, applying contraction we could have: 
\begin{equation}
    P(T^2)=\mathcal{C}(P(T^n))
\end{equation}
For $P(T^2)$, since Equ. \ref{T_2}, applying Lemma \ref{lemma2.1}, \begin{equation}
    P(T^2)=P(\mathcal{C}(T^n))=(P\times R)(T^2_{s})
\end{equation}
For its standard position $S(P(T))$ we have: 
\begin{equation}
    S(P(T))=(R^{-1}\times P^{-1}\times P)(T)=R^{-1}(T)=S(T)
\end{equation}
To simplify, for a rotation operation $P$ acting on an even high order tensor $T$, 
\begin{equation}
    S(P(T))=S(T).
\end{equation}
This satisfy the definition of rotation invariant. Therefore, for even cases, the standard position $S(T)$ is a rotation invariant.

Suppose $T$ has odd order. 

Let $\mathcal{C}$ be the sequence of contraction
along the first two axes such that $C(T) = T^3$, where $T^3$ is
a third-order tensor as described in the algorithm. 

\begin{equation}
    T^3=\mathcal{C}(T^n)
\end{equation}
Let $V_1, V_2, V_3$ be vectors of contraction operation on $T^3$ via different axes, \emph{i.e.}, 
\begin{equation}
    V_1=C(2,3)(T^3) \\ V_2=C(1,3)(T^3) \\ V_3=C(1,2)(T^3)
\end{equation}
Based on $S$, we have
\begin{equation}
    [V_1\;\;V_2\;\;V_3]=R_1\times U_1
    \label{thm2:equ1}
\end{equation}
In this case, 
\begin{equation}
    S(T^n)=R_1^{-1}(T^n)
    \label{thm2:equ5}
\end{equation}
Consider any rotation operation $P$ acting on $T^n$. We have,
\begin{equation}
    P(T^3)=P(\mathcal{C}(T^n))
\end{equation}
Using QR-factorization, 
\begin{equation}
    [P\times V_1\;\;P\times V_2\;\;P\times V_3]=R_2\times U_2
    \label{thm2:equ2}
\end{equation}
The standard position of $P(T^n)$ will be defined as:
\begin{equation}
    S(P(T^n))=R_2^{-1}(P(T^n))
    \label{thm2:equ6}
\end{equation}
Using Remark \ref{lemma2.2}, we could obtain
\begin{equation}
    \mathcal{C}(2,3)(P(T^3))=P\times \mathcal{C}(2,3)(T^3)=P\times V_1
\end{equation}
Considering $V_2$ and $V_3$, for the same reason, we could know that
\begin{equation}
    [P\times V_1\;\;P\times V_2\;\;P\times V_3]=P\times [V_1\;\;V_2\;\;V_3]
    \label{thm2:equ3}
\end{equation}
By reorganizing \ref{thm2:equ1}, \ref{thm2:equ2}, and \ref{thm2:equ3}, 
\begin{equation}
    [V_1\;\;V_2\;\;V_3] = R_1\times U_1=P^{-1}\times R_2\times U_2
\end{equation}
Since QR-factorization is unique \citep{golub2012matrix}, we should have that $U_1=U_2$. Therefore, 
\begin{equation}
    R_2 = P\times R_1
    \label{thm2:equ4}
\end{equation}
Plugging \ref{thm2:equ4} into \ref{thm2:equ6}, comparing the result of \ref{thm2:equ5} we have:
\begin{equation}
    S(P(T^n))=R_2^{-1}(P(T^n))=(R_1^{-1}\times P^{-1} \times P)(T^n) = R_1^{-1}(T^n)=S(T^n)
\end{equation}
Here, we shown that given any rotation operation $P$ on $T^n$ ($n$ is odd). By position standarization algorithm $S$, we will always have:
\begin{equation}
    S(P(T^n))=S(T^n)
\end{equation}
This satisfy the definition of rotation invariant. Therefore, for odd cases, the standard position $S(T)$ is a rotation invariant. 

Combining with the proof on even and odd cases, we have shown $S$ is rotation-invariant. 
\end{proof}
\bibliography{refs.bib}

\begin{thebibliography}{10}
\expandafter\ifx\csname url\endcsname\relax
  \def\url#1{\texttt{#1}}\fi
\expandafter\ifx\csname urlprefix\endcsname\relax\def\urlprefix{URL }\fi
\expandafter\ifx\csname href\endcsname\relax
  \def\href#1#2{#2} \def\path#1{#1}\fi

\bibitem{huang2020data}
Z.~Huang, Y.~Tian, C.~Li, G.~Lin, L.~Wu, Y.~Wang, H.~Jiang, Data-driven
  automated discovery of variational laws hidden in physical systems, Journal
  of the Mechanics and Physics of Solids (2020) 103871.

\bibitem{raissi2017physics}
M.~Raissi, P.~Perdikaris, G.~E. Karniadakis, Physics informed deep learning
  (part i): Data-driven solutions of nonlinear partial differential equations,
  arXiv preprint arXiv:1711.10561.

\bibitem{carleo2019machine}
G.~Carleo, I.~Cirac, K.~Cranmer, L.~Daudet, M.~Schuld, N.~Tishby,
  L.~Vogt-Maranto, L.~Zdeborov{\'a}, Machine learning and the physical
  sciences, Reviews of Modern Physics 91~(4) (2019) 045002.

\bibitem{kutz2017deep}
J.~N. Kutz, Deep learning in fluid dynamics, Journal of Fluid Mechanics 814
  (2017) 1--4.

\bibitem{wang2017physics}
J.-X. Wang, J.-L. Wu, H.~Xiao, Physics-informed machine learning approach for
  reconstructing reynolds stress modeling discrepancies based on dns data,
  Physical Review Fluids 2~(3) (2017) 034603.

\bibitem{li2019accelerating}
Y.~Li, T.~Zhang, S.~Sun, X.~Gao, Accelerating flash calculation through deep
  learning methods, Journal of Computational Physics 394 (2019) 153--165.

\bibitem{durbin2018some}
P.~A. Durbin, Some recent developments in turbulence closure modeling, Annual
  Review of Fluid Mechanics 50 (2018) 77--103.

\bibitem{ling2015evaluation}
J.~Ling, J.~Templeton, Evaluation of machine learning algorithms for prediction
  of regions of high reynolds averaged navier stokes uncertainty, Physics of
  Fluids 27~(8) (2015) 085103.

\bibitem{milano2002neural}
M.~Milano, P.~Koumoutsakos, Neural network modeling for near wall turbulent
  flow, Journal of Computational Physics 182~(1) (2002) 1--26.

\bibitem{zhang2015machine}
Z.~J. Zhang, K.~Duraisamy, Machine learning methods for data-driven turbulence
  modeling, in: 22nd AIAA Computational Fluid Dynamics Conference, 2015, p.
  2460.

\bibitem{beck2019deep}
A.~Beck, D.~Flad, C.-D. Munz, Deep neural networks for data-driven les closure
  models, Journal of Computational Physics 398 (2019) 108910.

\bibitem{chen2018neural}
T.~Q. Chen, Y.~Rubanova, J.~Bettencourt, D.~K. Duvenaud, Neural ordinary
  differential equations, in: Advances in neural information processing
  systems, 2018, pp. 6571--6583.

\bibitem{mase2009continuum}
G.~T. Mase, R.~E. Smelser, G.~E. Mase, Continuum mechanics for engineers, CRC
  press, 2009.

\bibitem{pope2001turbulent}
S.~B. Pope, Turbulent flows (2001).

\bibitem{pope1975more}
S.~Pope, A more general effective-viscosity hypothesis, Journal of Fluid
  Mechanics 72~(2) (1975) 331--340.

\bibitem{ling2016machine}
J.~Ling, R.~Jones, J.~Templeton, Machine learning strategies for systems with
  invariance properties, Journal of Computational Physics 318 (2016) 22--35.

\bibitem{ling2016reynolds}
J.~Ling, A.~Kurzawski, J.~Templeton, Reynolds averaged turbulence modelling
  using deep neural networks with embedded invariance, Journal of Fluid
  Mechanics 807 (2016) 155--166.

\bibitem{johnson2016handbook}
R.~W. Johnson, Handbook of fluid dynamics, Crc Press, 2016.

\bibitem{smith1965isotropic}
G.~Smith, On isotropic integrity bases, Archive for rational mechanics and
  analysis 18~(4) (1965) 282--292.

\bibitem{mohri2018foundations}
M.~Mohri, A.~Rostamizadeh, A.~Talwalkar, Foundations of machine learning, MIT
  press, 2018.

\bibitem{saxe2019information}
A.~M. Saxe, Y.~Bansal, J.~Dapello, M.~Advani, A.~Kolchinsky, B.~D. Tracey,
  D.~D. Cox, On the information bottleneck theory of deep learning, Journal of
  Statistical Mechanics: Theory and Experiment 2019~(12) (2019) 124020.

\bibitem{hornik1989multilayer}
K.~Hornik, M.~Stinchcombe, H.~White, et~al., Multilayer feedforward networks
  are universal approximators., Neural networks 2~(5) (1989) 359--366.

\bibitem{muller1999application}
S.~M{\"u}ller, M.~Milano, P.~Koumoutsakos, Application of machine learning
  algorithms to flow modeling and optimization, Annual Research Briefs (1999)
  169--178.

\bibitem{weatheritt2017comparative}
J.~Weatheritt, R.~D. Sandberg, J.~Ling, G.~Saez, J.~Bodart, A comparative study
  of contrasting machine learning frameworks applied to rans modeling of jets
  in crossflow, in: ASME Turbo Expo 2017: Turbomachinery Technical Conference
  and Exposition, American Society of Mechanical Engineers Digital Collection,
  2017.

\bibitem{qin2019data}
T.~Qin, K.~Wu, D.~Xiu, Data driven governing equations approximation using deep
  neural networks, Journal of Computational Physics 395 (2019) 620--635.

\bibitem{han2019solving}
J.~Han, L.~Zhang, E.~Weinan, Solving many-electron schr{\"o}dinger equation
  using deep neural networks, Journal of Computational Physics 399 (2019)
  108929.

\bibitem{lecun2015deep}
Y.~LeCun, Y.~Bengio, G.~Hinton, Deep learning, nature 521~(7553) (2015) 436.

\bibitem{cohen2016group}
T.~Cohen, M.~Welling, Group equivariant convolutional networks, in:
  International conference on machine learning, 2016, pp. 2990--2999.

\bibitem{esteves2020theoretical}
C.~Esteves, Theoretical aspects of group equivariant neural networks, arXiv
  preprint arXiv:2004.05154.

\bibitem{esteves2019equivariant}
C.~Esteves, Y.~Xu, C.~Allen-Blanchette, K.~Daniilidis, Equivariant multi-view
  networks, in: Proceedings of the IEEE International Conference on Computer
  Vision, 2019, pp. 1568--1577.

\bibitem{weiler2018learning}
M.~Weiler, F.~A. Hamprecht, M.~Storath, Learning steerable filters for rotation
  equivariant cnns, in: Proceedings of the IEEE Conference on Computer Vision
  and Pattern Recognition, 2018, pp. 849--858.

\bibitem{cheng2018rotdcf}
X.~Cheng, Q.~Qiu, R.~Calderbank, G.~Sapiro, Rotdcf: Decomposition of
  convolutional filters for rotation-equivariant deep networks, arXiv preprint
  arXiv:1805.06846.

\bibitem{finzi2020generalizing}
M.~Finzi, S.~Stanton, P.~Izmailov, A.~G. Wilson, Generalizing convolutional
  neural networks for equivariance to lie groups on arbitrary continuous data,
  arXiv preprint arXiv:2002.12880.

\bibitem{gao2019rotation}
L.~Gao, H.~Li, Z.~Lu, G.~Lin, Rotation-equivariant convolutional neural network
  ensembles in image processing, in: Adjunct Proceedings of the 2019 ACM
  International Joint Conference on Pervasive and Ubiquitous Computing and
  Proceedings of the 2019 ACM International Symposium on Wearable Computers,
  2019, pp. 551--557.

\bibitem{foley1996computer}
J.~D. Foley, F.~D. Van, A.~Van~Dam, S.~K. Feiner, J.~F. Hughes, J.~HUGHES,
  E.~ANGEL, Computer graphics: principles and practice, Vol. 12110,
  Addison-Wesley Professional, 1996.

\bibitem{zhou2019continuity}
Y.~Zhou, C.~Barnes, J.~Lu, J.~Yang, H.~Li, On the continuity of rotation
  representations in neural networks, in: Proceedings of the IEEE Conference on
  Computer Vision and Pattern Recognition, 2019, pp. 5745--5753.

\bibitem{curtis2012abstract}
M.~L. Curtis, Abstract linear algebra, Springer Science \& Business Media,
  2012.

\bibitem{bishop2006pattern}
C.~M. Bishop, Pattern recognition and machine learning, springer, 2006.

\bibitem{tao2005supervised}
D.~Tao, X.~Li, W.~Hu, S.~Maybank, X.~Wu, Supervised tensor learning, in: Fifth
  IEEE International Conference on Data Mining (ICDM'05), IEEE, 2005, pp.
  8--pp.

\bibitem{specht1991general}
D.~F. Specht, et~al., A general regression neural network, IEEE transactions on
  neural networks 2~(6) (1991) 568--576.

\bibitem{liaw2002classification}
A.~Liaw, M.~Wiener, et~al., Classification and regression by randomforest, R
  news 2~(3) (2002) 18--22.

\bibitem{pinter2010book}
C.~C. Pinter, A book of abstract algebra, Courier Corporation, 2010.

\bibitem{batchelor2000introduction}
C.~K. Batchelor, G.~Batchelor, An introduction to fluid dynamics, Cambridge
  university press, 2000.

\bibitem{kingma2014adam}
D.~P. Kingma, J.~Ba, Adam: A method for stochastic optimization, arXiv preprint
  arXiv:1412.6980.

\bibitem{pedregosa2011scikit}
F.~Pedregosa, G.~Varoquaux, A.~Gramfort, V.~Michel, B.~Thirion, O.~Grisel,
  M.~Blondel, P.~Prettenhofer, R.~Weiss, V.~Dubourg, et~al., Scikit-learn:
  Machine learning in python, the Journal of machine Learning research 12
  (2011) 2825--2830.

\bibitem{kosovic1997subgrid}
B.~KOSOVI{\'C}, Subgrid-scale modelling for the large-eddy simulation of
  high-reynolds-number boundary layers, Journal of Fluid Mechanics 336 (1997)
  151--182.

\bibitem{pitsch2006large}
H.~Pitsch, Large-eddy simulation of turbulent combustion, Annu. Rev. Fluid
  Mech. 38 (2006) 453--482.

\bibitem{matai2018flow}
R.~Matai, Les of flow over bumps and machine learning augmented turbulence
  modeling.

\bibitem{nye1985physical}
J.~F. Nye, et~al., Physical properties of crystals: their representation by
  tensors and matrices, Oxford university press, 1985.

\bibitem{yang2019predicting}
K.~Yang, X.~Xu, B.~Yang, B.~Cook, H.~Ramos, N.~A. Krishnan, M.~M. Smedskjaer,
  C.~Hoover, M.~Bauchy, predicting the young’s modulus of silicate glasses
  using high-throughput molecular dynamics simulations and machine learning,
  Scientific reports 9~(1) (2019) 1--11.

\bibitem{liu2019deep}
Z.~Liu, C.~Wu, M.~Koishi, A deep material network for multiscale topology
  learning and accelerated nonlinear modeling of heterogeneous materials,
  Computer Methods in Applied Mechanics and Engineering 345 (2019) 1138--1168.

\bibitem{walpole1984fourth}
L.~Walpole, Fourth-rank tensors of the thirty-two crystal classes:
  multiplication tables, Proceedings of the Royal Society of London. A.
  Mathematical and Physical Sciences 391~(1800) (1984) 149--179.

\bibitem{golub2012matrix}
G.~H. Golub, C.~F. Van~Loan, Matrix computations, Vol.~3, JHU press, 2012.

\end{thebibliography}

\end{document}